\pdfoutput=1
\documentclass{article}
\usepackage{iclr2027_conference,times}

\usepackage{hyperref}
\usepackage{url}
\usepackage{amsmath}
\usepackage{amssymb}
\usepackage{mathtools}
\usepackage{amsthm}
\usepackage{graphicx}
\usepackage{booktabs}
\usepackage{algorithm}
\usepackage{algorithmic}
\usepackage[capitalize,noabbrev]{cleveref}
\hypersetup{hidelinks}
\crefname{assumption}{Assumption}{Assumptions}
\Crefname{assumption}{Assumption}{Assumptions}

\theoremstyle{plain}
\newtheorem{theorem}{Theorem}
\newtheorem{lemma}[theorem]{Lemma}
\newtheorem{corollary}[theorem]{Corollary}
\theoremstyle{definition}
\newtheorem{definition}[theorem]{Definition}
\newtheorem{assumption}[theorem]{Assumption}
\theoremstyle{remark}

\theoremstyle{plain}
\newtheorem{maintheorem}{Theorem}
\newtheorem{proposition}{Proposition}
\crefname{maintheorem}{Theorem}{Theorems}
\Crefname{maintheorem}{Theorem}{Theorems}

\title{Neural ODEs Meet Concurrent Learning: \\ Stable Online Learning \\ with Lyapunov Guarantees}

\author{Omkar Sudhir Patil \\
Division of Electrical and Computer Engineering \\
Louisiana State University \\
Baton Rouge, LA 70803, USA \\
\texttt{opatil1@lsu.edu}
}

\iclrfinalcopy

\begin{document}

\maketitle
\lhead{Preprint. Under review as a conference paper at ICLR 2027.}

\begin{abstract}
Neural ODEs learn dynamics from trajectory losses, but their adjoint gradients lack the regressor-times-parameter-error structure on which Lyapunov analyses of online adaptation rest, so training on streaming data comes without stability guarantees. We show that this structure is in fact present: the adjoint gradient decomposes exactly into a positive semi-definite trajectory operator acting on the parameter error plus a nonlinear perturbation with explicit, horizon-dependent bounds. A quadratic Lyapunov function then certifies online Neural ODE training over sliding windows under computable gain and horizon conditions, and the same certificate extends to stored data: its drift branch recovers concurrent learning, and its trajectory branch yields NODE-CL, a stored-segment Gauss--Newton method built on batched forward sensitivities that needs no state-derivative estimates. On four DeepMind Control Suite domains, NODE-CL attains the lowest median prediction error on three under velocity measurement noise, where observer-based concurrent learning degrades by up to $8\times$; with clean measurements it is best on the pendulum and within a factor of $1.6$ of the best stored-data baseline on the cartpole and reacher.
\end{abstract}

\section{Introduction}
\label{sec:introduction}

Two communities learn unknown dynamics from data with neural networks and rarely talk to each other. In machine learning, a neural ordinary differential equation (Neural ODE) is fitted by minimizing a loss over predicted trajectories with adjoint gradients \citep{Chen2018}; the loss is informative because a trajectory mismatch integrates the modeling error through the dynamics, but training is offline, with no guarantee about the parameter flow on streaming data. In adaptive control, the same network is updated online from the instantaneous prediction error, and a Lyapunov argument certifies that the estimation error stays bounded \citep{Sanner1992,Lewis1999,Patil2022}, because that residual is, to first order, a known regressor times the parameter error. Trajectory losses lack that structure: the predicted trajectory depends on the parameters through an entire solve, so the adjoint gradient is not a regressor times a parameter error, and least-squares gains inject an indefinite term into any Lyapunov derivative. Concurrent learning \citep{Chowdhary2011}, which replays stored samples, sits in between: certified, but tied to the instantaneous residual and hence to a state-derivative observer.

This paper proposes NODE-CL, a stored-segment Gauss--Newton method for online Neural ODE identification that needs no state-derivative estimates and comes with a Lyapunov certificate, and shows that the same certificate covers both communities' methods (\cref{fig:overview}). The result that makes this possible is a trajectory-error decomposition, under which the adjoint gradient decomposes into a positive semi-definite trajectory operator acting on the parameter estimation error plus a perturbation with explicit, window-length dependent constants. A quadratic Lyapunov function then suffices: online Neural ODE training over sliding windows is certified under computable gain and horizon conditions, the window can be replaced by stored segments, and a Gauss--Newton gain from batched forward sensitivities can be added without leaving the certificate. The drift-based formulation is the foundation: with a least-squares gain driven by the regressor of the gradient decomposition, the indefinite gain-dynamics term is dominated by the gradient term, and the stored-set version of that law contains the stack structure of concurrent learning as its constant-gain instance.

\begin{figure}[t]
\centering
\includegraphics[width=0.9\textwidth]{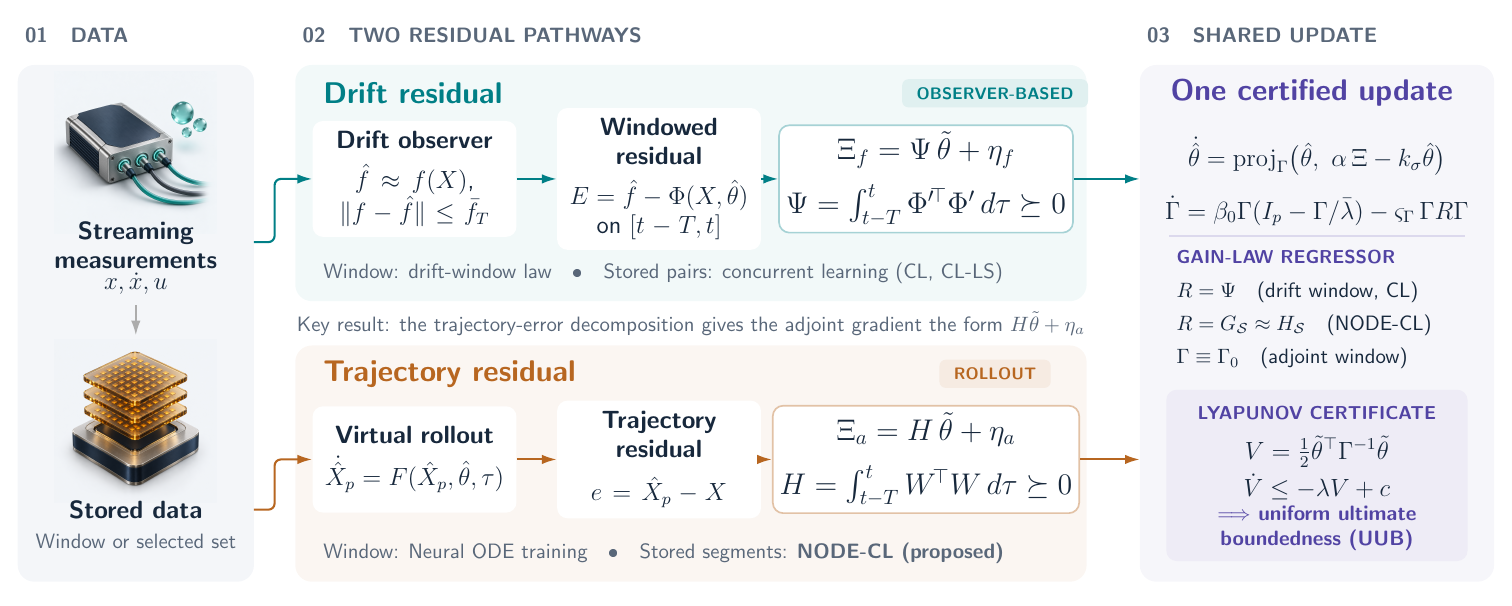}
\caption{Streaming data (01) feed two residual pathways (02), an instantaneous drift residual formed with an observer and a trajectory residual formed by rolling the model out from stored initial states; each gradient decomposes into a positive semi-definite operator on the parameter error plus a bounded perturbation, so both pathways share one certified update (03).}
\label{fig:overview}
\end{figure}

\textbf{Contributions.} First, NODE-CL: on four DeepMind Control Suite domains it is the most robust of the compared methods to velocity measurement noise and, with clean data, is best on the pendulum and within a factor of $1.6$ of the best stored-data baseline on the cartpole and reacher, against concurrent-learning baselines derived from \citet{Hart2025} (CL-LS and CL) and controlled Adam replay, with horizon, observer, and gain ablations. Second, the trajectory-error decomposition, which gives the adjoint gradient of a nonlinear-in-parameter trajectory model the regressor form needed for Lyapunov analysis, with explicit horizon-dependent constants. Third, two further certified online methods that the same analysis yields, drift-window and adjoint-window identification. Computational costs and limitations are reported in \cref{app:sim_details}.

\textbf{Related work.} Neural ODEs \citep{Chen2018,Dupont2019,Rubanova2019,Kidger2020,Massaroli2020,Kidger2022} are commonly trained offline with adjoint or checkpointed gradients \citep{Zhuang2020}, and long horizons are handled by multiple shooting \citep{Turan2022,Massaroli2021} or by regularizing the flow \citep{Finlay2020}. Existing Lyapunov-based neural adaptive laws \citep{Sanner1992,Lewis1999,Patil2022,Joshi2019,Sun2022} use instantaneous residuals; composite adaptation and concurrent learning add recorded data \citep{SlotineLi1989,Chowdhary2011,Kamalapurkar2017}. Least-squares DNN laws appear in \citet{Makumi2025,Hart2025}. Our drift branch adds certified gain bounds and a gain-metric projection, while the trajectory branch removes the derivative observer.

The NODE-CL gain recursion is related to covariance updates in extended-Kalman-filter and recursive Gauss--Newton training \citep{SinghalWu1989,Puskorius1994,Haykin2001}; the new ingredient is the trajectory-error decomposition that places this update inside a Lyapunov certificate. This differs from certifying properties of the learned dynamics themselves, such as stability of the learned flow \citep{Kolter2019,Kang2021,Rodriguez2022,Luo2025} or of a constraint or attractor set \citep{White2023,Sochopoulos2024}, and from regret guarantees for instantaneous nonlinear adaptation \citep{Boffi2021,BoffiTu2021} and meta-learned adaptive controllers \citep{Richards2021,Shi2021,OConnell2022}; the certificate here concerns the online parameter flow. Classical identification and sparse regression use different model classes \citep{Ljung1999,Brunton2016}; the projection follows \citet{Lavretsky2013,Cai2006}. \cref{sec:problem} formulates the problem, \cref{sec:drift} develops the methods, and \cref{sec:simulations} reports the experiments.

\section{Problem Formulation}
\label{sec:problem}

Consider the second-order nonlinear system
\begin{equation}
\label{eq:system}
\ddot{x} = f(x, \dot{x}) + g(x, \dot{x}) u,
\end{equation}
where $x, \dot{x} \in \mathbb{R}^{n}$ are measurable states, $f: \mathbb{R}^{n} \times \mathbb{R}^{n} \to \mathbb{R}^{n}$ is an unknown drift to be identified, $g: \mathbb{R}^{n} \times \mathbb{R}^{n} \to \mathbb{R}^{n \times m}$ is a known input matrix, and $u \in \mathbb{R}^{m}$ is a known external input. Let $X \triangleq [x^{\top} \; \dot{x}^{\top}]^{\top} \in \mathbb{R}^{2n}$ and $B \triangleq [0_{n \times n} \; I_{n}]^{\top} \in \mathbb{R}^{2n \times n}$; we write $f(X)$ and $g(X)$ for $f(x,\dot{x})$ and $g(x,\dot{x})$. The following assumptions describe the identification experiment and the regularity of the system.

\begin{assumption}
\label{assm:bounded_trajectory}
The external input is bounded, $\left\Vert u(t) \right\Vert \leq \bar{u}$ for all $t \geq t_{0}$, and the resulting state trajectory is bounded, $\left\Vert X(t) \right\Vert \leq \bar{X}$ for all $t \geq t_{0}$ (standard in identification, where boundedness is enforced by a stabilizing baseline input, as in \cref{sec:simulations}).
\end{assumption}

\begin{assumption}
\label{assm:f_smooth}
The drift $f$ is continuously differentiable, and $g$ is twice continuously differentiable. Consequently, on $\{X : \left\Vert X \right\Vert \leq \bar{X}\}$ there exist constants $c_{fx}, c_{fv}, \bar{f}, \bar{g} > 0$ such that $\left\Vert \partial f / \partial x \right\Vert \leq c_{fx}$, $\left\Vert \partial f / \partial \dot{x} \right\Vert \leq c_{fv}$, $\left\Vert f \right\Vert \leq \bar{f}$, and $\left\Vert g \right\Vert \leq \bar{g}$.
\end{assumption}

Let $\Phi: \mathbb{R}^{2n} \times \mathbb{R}^{p} \to \mathbb{R}^{n}$ denote a DNN with parameters $\theta \in \mathbb{R}^{p}$ and $\Phi'(X,\theta) \triangleq \partial \Phi / \partial \theta \in \mathbb{R}^{n \times p}$. Let $\Theta \triangleq \{\theta \in \mathbb{R}^{p} : \left\Vert \theta \right\Vert \leq \bar{\theta}\}$ denote the admissible parameter set. Fix a margin $\varpi > 0$ and define the compact convex domain $\Omega \triangleq \{X \in \mathbb{R}^{2n} : \left\Vert X \right\Vert \leq \bar{X} + \varpi\}$, which contains the trajectory and a tube in which the virtual predictions below remain. Since $\Phi$ is a composition of smooth layers, it is twice continuously differentiable on any compact set, so there exist constants $\bar{\upsilon}, c_{\theta\theta}, c_{X\theta} > 0$ such that, on $\Omega \times \Theta$, $\left\Vert \Phi'(X,\theta) \right\Vert \leq \bar{\upsilon}$, $\left\Vert \partial^{2}\Phi/\partial \theta^{2} \right\Vert \leq c_{\theta\theta}$, and $\left\Vert \partial^{2}\Phi/\partial X \partial \theta \right\Vert \leq c_{X\theta}$.

Let $\mu$ denote a finite Borel measure on $\Omega$, and define the ideal parameters as $\theta^{*} = \arg\min_{\theta \in \Theta} \int_{\Omega} \left\Vert f(X) - \Phi(X, \theta) \right\Vert^{2} \, d\mu(X) + \rho \left\Vert \theta \right\Vert^{2}$ with a regularization weight $\rho > 0$.\footnote{The reader is referred to \citet{Hart2025arxiv} for identifiability conditions under which this regularized minimizer is unique.} The drift is then modeled as
\begin{equation}
\label{eq:ufap}
f(X) = \Phi(X, \theta^{*}) + \varepsilon(X), \qquad \max_{X \in \Omega} \left\Vert \varepsilon(X) \right\Vert \leq \bar{\varepsilon},
\end{equation}
where $\bar{\varepsilon}$ can be made small by an expressive architecture and a large enough $\bar{\theta}$; every ultimate bound below grows with $\bar{\varepsilon}$.

The objective is an online law for $\hat{\theta}(t)$ under which the estimation error $\tilde{\theta}(t) \triangleq \theta^{*} - \hat{\theta}(t)$ is uniformly ultimately bounded, with an exponential rate into a residual set determined by the leakage, approximation, and observer errors, which shrinks under a window persistence-of-excitation condition stated in the appendix. \cref{sec:drift} develops the adaptation laws that achieve this objective.
\section{Method: From Local Residuals to NODE-CL}
\label{sec:drift}

We develop the method in three stages, each building on the previous one. First, we define a drift residual, the mismatch between an observer-based drift estimate and the model over a sliding window, and derive a gradient-based estimator with a certified least-squares gain. Second, we replace the drift residual by a trajectory residual, the mismatch between a model rollout and the measured trajectory, which removes the observer, and show that its adjoint gradient admits the same certified estimator. Third, we apply concurrent learning to the trajectory regressor: NODE-CL replays stored trajectory segments instead of a sliding window and rescales their directions with a Gauss--Newton gain. \cref{sec:stage1} presents the first stage.

\subsection{Stage I: proposed drift-window method}
\label{sec:stage1}

We introduce the drift-window method to convert streaming state and input measurements into approximate drift labels and to adapt the model over a sliding window of those labels. Let $\hat{f}$ denote a drift estimate and $\tilde{f} \triangleq f(X)-\hat{f}$ its error. Dynamic state-derivative estimators are standard in concurrent learning \citep{Kamalapurkar2017,Hart2025}. Throughout, the estimate is kept by a projection (introduced in \cref{sec:gain}) inside $\Theta_{\epsilon} \triangleq \{\theta : \left\Vert \theta \right\Vert \leq \bar{\theta} + \epsilon_{\theta}\}$, the inflation of $\Theta$ by a margin $\epsilon_{\theta} > 0$, so $\left\Vert \tilde{\theta}(t) \right\Vert \leq \bar{\tilde{\theta}} \triangleq 2\bar{\theta} + \epsilon_{\theta}$; the constants $\bar{\upsilon}, c_{\theta\theta}, c_{X\theta}$ are taken so that the bounds of \cref{sec:problem} hold on $\Omega \times \Theta_{\epsilon}$. The analysis below requires only an estimator error envelope, rather than a particular construction. For completeness, \cref{app:observer_proof} specifies the observer used in the experiments and proves the bound below. For a window length $T>0$ and $t\geq t_{0}+T$, the resulting bound on the observer error over the window is
\begin{equation}
\sup_{\tau\in[t-T,t]}\left\Vert \tilde{f}(\tau) \right\Vert \leq \bar{f}_{T}(t) \triangleq \left\Vert z_{o}(t_{0}) \right\Vert e^{-k_{o}(t-T-t_{0})}+\bar{z}_{o}. \label{eq:observer_window_bound}
\end{equation}
where $z_{o}$ is the observer error state, $k_{o} > 0$ its decay rate, and $\bar{z}_{o}$ its steady-state residual, which decreases with the observer gain (all given explicitly in \cref{app:observer_proof}).
\cref{sec:windowloss} defines the windowed loss and decomposes its gradient, \cref{sec:gain} introduces the gain and projection, and \cref{sec:law} states the adaptation law and its guarantee.

\subsubsection{Window loss and gradient decomposition}
\label{sec:windowloss}

To make the update less sensitive to the error of any single sample, we accumulate the residual over a sliding window rather than at a single instant, and we show that the gradient of this windowed loss has the regressor form needed for the Lyapunov analysis. Fix a window length $T > 0$. For $t \geq t_{0} + T$, the windowed drift prediction error and loss are defined as
\begin{align}
E(\tau; t) &\triangleq \hat{f}(\tau) - \Phi(X(\tau), \hat{\theta}(t)), \label{eq:E_def}\\
\mathcal{E}_{f}(t) &\triangleq \tfrac{1}{2}\int_{t-T}^{t} \left\Vert E(\tau; t) \right\Vert^{2}\, d\tau, \label{eq:Ef_def}
\end{align}
where the pairs $(X(\tau), \hat{f}(\tau))$ are buffered and $\hat{\theta}(t)$ is the current estimate. Since only $\Phi$ depends on $\hat{\theta}$, the negative gradient of \cref{eq:Ef_def} is
\begin{equation}
\label{eq:Xif_def}
\Xi_{f}(t) \triangleq -\nabla_{\hat{\theta}}\mathcal{E}_{f}(t) = \int_{t-T}^{t} \Phi'(X(\tau), \hat{\theta}(t))^{\top} E(\tau; t)\, d\tau.
\end{equation}
Expanding $\Phi(X(\tau), \theta^{*})$ about $\hat{\theta}(t)$ in each sample gives the following decomposition of $\Xi_{f}$, which is the form every Lyapunov argument below relies on.

\begin{proposition}
\label{lem:drift_decomp}
Let \cref{assm:bounded_trajectory,assm:f_smooth} hold, and let $\hat{\theta}(t) \in \Theta_{\epsilon}$. Then, for $t \geq t_{0} + T$,
\begin{equation}
\label{eq:Xif_decomp}
\Xi_{f}(t) = \Psi(t)\, \tilde{\theta}(t) + \eta_{f}(t),
\end{equation}
where the windowed regressor
\begin{equation}
\label{eq:Psi_def}
\Psi(t) \triangleq \int_{t-T}^{t} \Phi'(X(\tau), \hat{\theta}(t))^{\top} \Phi'(X(\tau), \hat{\theta}(t))\, d\tau
\end{equation}
is positive semi-definite with $\left\Vert \Psi(t) \right\Vert \leq T\bar{\upsilon}^{2}$, and the perturbation satisfies
\begin{equation}
\label{eq:etaf_bound}
\left\Vert \eta_{f}(t) \right\Vert \leq c_{\eta 1}\left\Vert \tilde{\theta}(t) \right\Vert^{2} + c_{\eta 2}\,\bar{f}_{T}(t) + c_{\eta 3},
\end{equation}
with $c_{\eta 1} \triangleq \tfrac{1}{2} T \bar{\upsilon}\, c_{\theta\theta}$, $c_{\eta 2} \triangleq T\bar{\upsilon}$, and $c_{\eta 3} \triangleq T\bar{\upsilon}\,\bar{\varepsilon}$.
\end{proposition}

The proof is given in \cref{app:drift_decomp_proof}; the three terms in \cref{eq:etaf_bound} are the Taylor remainder, the observer error, and the approximation error, each scaled by $T$. \cref{sec:gain} introduces the gain and projection used in the update law.

\subsubsection{Least-squares gain and projection}
\label{sec:gain}

The adaptation law also requires a least-squares gain for excited directions and a projection that keeps the estimate inside $\Theta_{\epsilon}$, where the smoothness constants of $\Phi$ hold. The gain evolves as
\begin{equation}
\label{eq:Gamma_update}
\dot{\Gamma} = \beta_{0}\, \Gamma\left(I_{p} - \Gamma/\bar{\lambda}\right) - \varsigma_{\Gamma}(t)\, \Gamma\, \Psi(t)\, \Gamma, \qquad \Gamma(t_{0}) = \Gamma_{0},
\end{equation}
where $\beta_{0}>0$ is the forgetting rate and $\bar{\lambda}>0$ is a cap. The Lipschitz gate $\varsigma_{\Gamma}(t)=\varsigma(\lambda_{\min}(\Gamma(t)))$ is zero below the floor $\underline{\lambda}$, one above $2\underline{\lambda}$, and linear between them, with $0<\underline{\lambda}<\bar{\lambda}/2$; set $\Psi=0$ during the initial window. Driving the gain by the same regressor as \cref{eq:Xif_decomp} preserves $\underline{\lambda}I_p\preceq\Gamma(t)\preceq\bar{\lambda}I_p$ (proved in \cref{app:gamma_proof}). The gain-metric projection $\mathrm{proj}_{\Gamma}(\theta,y)$ equals $\Gamma y$ inside $\Theta$ and, in the margin between $\Theta$ and $\Theta_{\epsilon}$, removes the component of $\Gamma y$ that points outward (defined in \cref{app:proj_proof}); it keeps $\hat{\theta}(t)\in\Theta_\epsilon$ and satisfies, for every $\theta^{*}\in\Theta$,
\begin{equation}
\label{eq:proj_prop}
\tilde{\theta}^{\top} \Gamma^{-1} \mathrm{proj}_{\Gamma}(\hat{\theta}, y) \geq \tilde{\theta}^{\top} y
\end{equation}
(proved in \cref{app:proj_proof}). This metric compatibility is the property used in every Lyapunov proof. \cref{sec:law} states the adaptation law and its guarantee.

\subsubsection{Adaptation law and guarantee}
\label{sec:law}

The drift-based windowed adaptation law is designed as
\begin{equation}
\label{eq:adapt_law}
\dot{\hat{\theta}} = \mathrm{proj}_{\Gamma}\!\left(\hat{\theta},\; \alpha\, \Xi_{f}(t) - k_{\sigma}\, \hat{\theta} \right),
\end{equation}
where $\alpha>0$ is the learning gain and $k_{\sigma}>0$ is a sigma-modification gain; set $\Xi_f=0$ during the initial window. This is a projected, least-squares preconditioned gradient flow on \cref{eq:Ef_def}. Existence and uniqueness of solutions are addressed in \cref{app:proj_proof}. The following proposition certifies this law.

\begin{proposition}
\label{thm:drift}
Let \cref{assm:bounded_trajectory,assm:f_smooth} hold, let $\hat{\theta}(t_{0}) \in \Theta_{\epsilon}$, let $\underline{\lambda} I_{p} \preceq \Gamma_{0} \preceq \bar{\lambda} I_{p}$, and let the gains satisfy $\alpha > \tfrac{1}{2}$ and $k_{\sigma} > \alpha\, T\, \bar{\upsilon}\, c_{\theta\theta}\, \bar{\tilde{\theta}}$. Suppose the drift estimate satisfies \cref{eq:observer_window_bound}. Then, under the gain dynamics in \cref{eq:Gamma_update} and the adaptation law in \cref{eq:adapt_law}, the parameter estimation error satisfies, for all $t \geq t_{1} \triangleq t_{0} + T$,
\begin{equation}
\label{eq:theta_bound}
\left\Vert \tilde{\theta}(t) \right\Vert \leq \sqrt{\tfrac{\bar{\lambda}}{\underline{\lambda}} \left\Vert \tilde{\theta}(t_{1}) \right\Vert^{2} e^{-\lambda_{d}(t - t_{1})} + \tfrac{2\bar{\lambda}}{\lambda_{d}}\, c_{d}(t_{1})},
\qquad
\limsup_{t \to \infty} \left\Vert \tilde{\theta}(t) \right\Vert \leq \sqrt{\tfrac{2\bar{\lambda}}{\lambda_{d}}\, c_{d}^{\infty}},
\end{equation}
where $k_{1} \triangleq \tfrac{k_{\sigma}}{2} - \alpha c_{\eta 1} \bar{\tilde{\theta}} > 0$, $\lambda_{d} \triangleq 2 k_{1} \underline{\lambda}$,
\begin{equation}
\label{eq:cd_def}
c_{d}(t) \triangleq k_{\sigma}\bar{\theta}^{2} + \tfrac{\alpha^{2}}{k_{\sigma}}\left(c_{\eta 2}\, \bar{f}_{T}(t) + c_{\eta 3}\right)^{2},
\qquad
c_{d}^{\infty} \triangleq k_{\sigma}\bar{\theta}^{2} + \tfrac{\alpha^{2}}{k_{\sigma}}\left(c_{\eta 2}\, \bar{z}_{o} + c_{\eta 3}\right)^{2}.
\end{equation}
\end{proposition}

The proof is given in \cref{app:drift_proof}: in $V = \tfrac{1}{2}\tilde{\theta}^{\top}\Gamma^{-1}\tilde{\theta}$ the gain dynamics produce $\tfrac{1}{2}\varsigma_{\Gamma}\tilde{\theta}^{\top}\Psi\tilde{\theta}$ and the projected gradient $-\alpha\tilde{\theta}^{\top}\Psi\tilde{\theta}$, so $\alpha > \tfrac{1}{2}$ leaves a sink; $k_{\sigma}\bar{\theta}^{2}$ is the usual sigma-modification tradeoff. Without excitation no law can localize $\theta$ beyond the projection set, and \cref{eq:theta_bound} reflects this (see the remark in \cref{app:drift_proof}).

Two further consequences are stated and proved in \cref{app:drift_proof}: the sink retained in the Lyapunov argument bounds the time-averaged first-order functional error without excitation, and it improves the parameter rate under a window persistence-of-excitation condition. The analysis also applies to noncontiguous stored drift pairs with uniformly bounded observer errors (\cref{app:stored_proof}); this stored-set result supplies the guarantee used by the adapted CL-LS baseline. \cref{sec:adjoint} replaces the drift residual by a trajectory residual.

\subsection{Stage II: proposed adjoint-window method}
\label{sec:adjoint}

We now replace the drift residual by a trajectory residual, the mismatch between a model rollout and the measured trajectory, which needs no derivative estimate. The difficulty is that its adjoint gradient has no immediate regressor form, and this section derives one: \cref{sec:virtual} defines the rollout and decomposes its adjoint gradient, and \cref{sec:certified} states the update and its guarantee.

\subsubsection{Virtual system, window loss, and adjoint gradient}
\label{sec:virtual}

Define the model vector field $F: \mathbb{R}^{2n} \times \mathbb{R}^{p} \times \mathbb{R} \to \mathbb{R}^{2n}$ as $F(\chi, \theta, \tau) \triangleq [\chi_{v}^{\top},\; (\Phi(\chi, \theta) + g(\chi)\, u_{\mathrm{rec}}(\tau))^{\top}]^{\top}$, where $\chi = [\chi_{x}^{\top}\; \chi_{v}^{\top}]^{\top}$ and $u_{\mathrm{rec}}$ is the recorded input. At time $t \geq t_{0}+T$, the virtual predicted trajectory $\hat{X}_{p}(\cdot; t): [t-T, t] \to \mathbb{R}^{2n}$ is generated by
\begin{equation}
\label{eq:virtual_system}
\frac{d \hat{X}_{p}}{d\tau} = F(\hat{X}_{p}, \hat{\theta}(t), \tau), \qquad \hat{X}_{p}(t-T; t) = X(t-T),
\end{equation}
and the trajectory window loss and prediction error are defined as
\begin{equation}
\label{eq:ET_def}
\mathcal{E}_{T}(t) \triangleq \tfrac{1}{2}\int_{t-T}^{t} \left\Vert e(\tau; t) \right\Vert^{2} d\tau, \qquad e(\tau;t) \triangleq \hat{X}_{p}(\tau; t) - X(\tau).
\end{equation}
Let $L_{F}, L_{FX} > 0$ bound $\left\Vert \partial F/\partial \chi \right\Vert$ and $\left\Vert \partial^{2} F/\partial \chi^{2} \right\Vert$ on $\Omega \times \Theta_{\epsilon}$, and let $\kappa_{T} \triangleq (e^{L_{F} T} - 1)/L_{F}$. The virtual trajectory remains in $\Omega$ under the window condition
\begin{equation}
\label{eq:containment_cond}
\kappa_{T}\left(\bar{\upsilon}\, \bar{\tilde{\theta}} + \bar{\varepsilon}\right) < \varpi,
\end{equation}
(proved in \cref{app:containment_proof}), which also gives $\left\Vert e(\tau;t) \right\Vert < \varpi$. 

By the adjoint sensitivity method \citep{Chen2018}, the gradient of \cref{eq:ET_def} with respect to $\hat{\theta}$ is
\begin{equation}
\label{eq:adjoint_gradient}
\begin{aligned}
\nabla_{\hat{\theta}}\, \mathcal{E}_{T}(t) &= \int_{t-T}^{t} \left(\frac{\partial F}{\partial \theta}(\hat{X}_{p}, \hat{\theta}(t), \tau)\right)^{\top} \nu(\tau)\, d\tau, \\
\frac{d\nu}{d\tau} &= -\left(\frac{\partial F}{\partial \chi}(\hat{X}_{p}, \hat{\theta}(t), \tau)\right)^{\top} \nu - e(\tau; t), \qquad \nu(t) = 0_{2n},
\end{aligned}
\end{equation}
where $\nu: [t-T, t] \to \mathbb{R}^{2n}$ is the adjoint state (derived in \cref{app:adjoint_proof}).

Since $\partial F/\partial \theta = B\, \Phi'(\hat{X}_{p}, \theta)$, partitioning $\nu = [\nu_{x}^{\top}\; \nu_{v}^{\top}]^{\top}$ gives the trajectory gradient signal
\begin{equation}
\label{eq:Xia_def}
\Xi_{a}(t) \triangleq -\nabla_{\hat{\theta}}\mathcal{E}_{T}(t) = -\int_{t-T}^{t} \Phi'(\hat{X}_{p}(\tau;t), \hat{\theta}(t))^{\top} \nu_{v}(\tau)\, d\tau.
\end{equation}

The main result of the paper, stated next, shows that the trajectory gradient has the same structure as the drift gradient: a positive semi-definite operator acting on the parameter error plus a bounded perturbation. It follows from an exact representation of the rollout error, obtained by variation of constants applied to a mean-value linearization, which yields a trajectory regressor $W(\tau)$ and a remainder $d(\tau)$ with $e=-(W\tilde{\theta}+d)$; their definitions and all constants are given in \cref{app:traj_decomp_proof}.

\begin{maintheorem}
\label{lem:traj_decomp}
Let \cref{assm:bounded_trajectory,assm:f_smooth} and \cref{eq:containment_cond} hold, and let $\hat{\theta}(t) \in \Theta_{\epsilon}$. Let $W$ and $d$ be the explicitly constructed quantities in \cref{app:traj_decomp_proof}. Then $e(\tau; t) = -\left(W(\tau)\, \tilde{\theta}(t) + d(\tau)\right)$ for all $\tau \in [t-T,t]$, and the adjoint gradient signal in \cref{eq:Xia_def} decomposes as
\begin{equation}
\label{eq:Xia_decomp}
\Xi_{a}(t) = H(t)\, \tilde{\theta}(t) + \eta_{a}(t),
\end{equation}
where $H(t) \triangleq \int_{t-T}^{t} W(\tau)^{\top} W(\tau)\, d\tau \succeq 0$ with $\left\Vert H(t) \right\Vert \leq T \bar{w}^{2}$, and
\begin{equation}
\label{eq:etaa_bound}
\left\Vert \eta_{a}(t) \right\Vert \leq b_{1}\left\Vert \tilde{\theta} \right\Vert^{2} + b_{2}\left\Vert \tilde{\theta} \right\Vert + b_{3},
\end{equation}
with $\bar{w} \triangleq \bar{\upsilon}\kappa_{T}$ and constants $b_{1}, b_{2}, b_{3} \geq 0$ defined in \cref{app:traj_decomp_proof}, where $b_{1} = O(T\kappa_{T}^{2})$, $b_{2}$ is proportional to $\bar{\varepsilon}$, and $b_{3}$ vanishes with $\bar{\varepsilon}$.
\end{maintheorem}

The proof is given in \cref{app:traj_decomp_proof}. Every constant in \cref{eq:etaa_bound} carries a factor $T\kappa_{T}$, so $b_{1}, b_{2}, b_{3} \to 0$ as $T \to 0$. No observer error appears because \cref{eq:ET_def} compares the rollout with measured states. The corresponding cost is the Gronwall growth of the constants with $T$. \cref{sec:certified} uses this decomposition to certify the online update.

\subsubsection{Certified online update}
\label{sec:certified}

The trajectory-based windowed adaptation law is designed as
\begin{equation}
\label{eq:adapt_law_a}
\dot{\hat{\theta}} = \mathrm{proj}_{\Gamma_{0}}\!\left(\hat{\theta},\; \alpha\, \Xi_{a}(t) - k_{\sigma}\, \hat{\theta}\right),
\end{equation}
with a constant symmetric positive definite gain $\Gamma_{0}$ satisfying $\underline{\lambda} I_{p} \preceq \Gamma_{0} \preceq \bar{\lambda} I_{p}$, and $\Xi_{a} \triangleq 0_{p}$ on $[t_{0}, t_{0}+T]$. Forming $H$ online would require extra solves, so the trajectory law uses a constant gain. This gain contributes no additional Lyapunov term, and therefore $\alpha > \tfrac{1}{2}$ is not required. The resulting guarantee is stated next.

\begin{proposition}
\label{thm:adjoint_uub}
Let \cref{assm:bounded_trajectory,assm:f_smooth} hold, let $\hat{\theta}(t_{0}) \in \Theta_{\epsilon}$, and let the window length and gains satisfy \cref{eq:containment_cond} and $k_{\sigma} > 2\alpha\left(b_{1}\bar{\tilde{\theta}} + b_{2}\right)$. Then, for the adaptation law in \cref{eq:adapt_law_a} with the adjoint gradient in \cref{eq:Xia_def}, the parameter estimation error satisfies, for all $t \geq t_{1} = t_{0} + T$,
\begin{equation}
\label{eq:theta_bound_a}
\left\Vert \tilde{\theta}(t) \right\Vert \leq \sqrt{\tfrac{\bar{\lambda}}{\underline{\lambda}} \left\Vert \tilde{\theta}(t_{1}) \right\Vert^{2} e^{-\lambda_{a}(t - t_{1})} + \tfrac{\bar{\lambda}\, c_{a}}{k_{1a}\, \underline{\lambda}}},
\end{equation}
where $k_{1a} \triangleq \tfrac{k_{\sigma}}{2} - \alpha\left(b_{1}\bar{\tilde{\theta}} + b_{2}\right) > 0$, $\lambda_{a} \triangleq 2 k_{1a} \underline{\lambda}$, and $c_{a} \triangleq k_{\sigma}\bar{\theta}^{2} + \alpha^{2} b_{3}^{2}/k_{\sigma}$.
\end{proposition}

The proof is given in \cref{app:adjoint_uub_proof}. Since $b_{3}$ vanishes with $\bar{\varepsilon}$, the residual $\alpha^{2}b_{3}^{2}/k_{\sigma}$ vanishes with the approximation error. Moreover, $b_{1}, b_{2} \to 0$ as $T \to 0$, so the gain condition can always be met by shortening the window. For $L_{F}T \leq 1$, it holds whenever $T \leq (k_{\sigma}/(2\alpha C_{b}))^{1/3}$ for a computable constant $C_{b}$ (\cref{app:adjoint_uub_proof}). \cref{sec:stack} extends the trajectory law to stored segments.

\subsection{Stage III: proposed stored-segment NODE-CL}
\label{sec:stack}

We now replace the sliding window of Stage II by stored trajectory segments, in the spirit of concurrent learning. A sliding window retains only the most recently visited region of the state space, and a constant gain does not compensate for the scale of the trajectory operator $H$, which grows as $T^{3}$ for short windows since $\left\Vert H \right\Vert \leq T\bar{w}^{2}$ and $\bar{w} = \bar{\upsilon}\kappa_{T} \approx \bar{\upsilon} T$. NODE-CL addresses these limitations by retaining selected trajectory segments and adapting the gain. Let $\mathcal{S}(t) = \{(X_{j}(\cdot), u_{j}(\cdot))\}_{j=1}^{N_{s}}$ be $N_{s}$ stored segments, each a measured state and input record over $[t_{j}, t_{j} + T_{s}]$ with $t_{j} + T_{s} \leq t$ and membership piecewise constant in $t$; let $\hat{X}_{p,j}$ solve \cref{eq:virtual_system} from $X_{j}(t_{j})$ with the current $\hat{\theta}(t)$, $e_{j} \triangleq \hat{X}_{p,j} - X_{j}$, and $S_{j} \triangleq \partial \hat{X}_{p,j}/\partial \hat{\theta}$. Define
\begin{equation}
\label{eq:stack_defs}
\Xi_{\mathcal{S}} \triangleq -\frac{1}{N_{s}}\sum_{j=1}^{N_{s}} \int_{t_{j}}^{t_{j}+T_{s}} S_{j}^{\top} e_{j}\, d\tau,
\qquad
G_{\mathcal{S}} \triangleq \frac{1}{N_{s}}\sum_{j=1}^{N_{s}} \int_{t_{j}}^{t_{j}+T_{s}} S_{j}^{\top} S_{j}\, d\tau,
\end{equation}
the negative gradient of the averaged segment loss and its Gauss--Newton matrix, both produced by one batched forward pass at cost $O(N_{s} N_{w} n^{2} p)$ per update ($N_{w}$ integration steps per segment). Applying the trajectory-error decomposition of \cref{lem:traj_decomp} segment by segment gives $\Xi_{\mathcal{S}} = H_{\mathcal{S}}\tilde{\theta} + \eta_{\mathcal{S}}$ with $H_{\mathcal{S}} \triangleq \tfrac{1}{N_{s}}\sum_{j=1}^{N_{s}}\int_{t_{j}}^{t_{j}+T_{s}} W_{j}^{\top} W_{j}\, d\tau \succeq 0$ and $\left\Vert \eta_{\mathcal{S}} \right\Vert \leq b_{1}\left\Vert \tilde{\theta} \right\Vert^{2} + b_{2}\left\Vert \tilde{\theta} \right\Vert + b_{3}$ at segment length $T_{s}$. The computable $G_{\mathcal{S}}$ differs from $H_{\mathcal{S}}$ only through $S_{j} - W_{j}$: as shown in \cref{app:stack_proof}, $\left\Vert G_{\mathcal{S}} - H_{\mathcal{S}} \right\Vert \leq g_{2}\left\Vert \tilde{\theta} \right\Vert^{2} + g_{1}\left\Vert \tilde{\theta} \right\Vert + g_{0}$ for constants $g_{0}, g_{1}, g_{2}$ built from those of \cref{lem:traj_decomp} at $T_{s}$.

The stored-segment law is $\dot{\hat{\theta}} = \mathrm{proj}_{\Gamma}(\hat{\theta}, \alpha\Xi_{\mathcal{S}} - k_{\sigma}\hat{\theta})$ with $\Gamma$ generated by \cref{eq:Gamma_update}, using $G_{\mathcal{S}}$ in place of $\Psi$. On the excited subspace, the forgetting update equilibrates at $\Gamma \approx \beta_{0} G_{\mathcal{S}}^{-1}$. The resulting projected Gauss--Newton flow is insensitive to the $T_{s}^{3}$ scale and can therefore use short segments.

\setcounter{algorithm}{2}
\begin{algorithm}[H]
\caption{Proposed NODE-CL: Stored-Segment Gauss--Newton Identification}
\label{alg:stack}
\begin{algorithmic}[1]
\REQUIRE Segment length $T_{s}$, budget $N_{s}$, gains $\alpha, k_{\sigma}, \beta_{0}, \bar{\lambda}, \underline{\lambda}$
\STATE Initialize $\hat{\theta}(t_{0}) \in \Theta_{\epsilon}$, $\Gamma(t_{0}) = \Gamma_{0}$, and segment set $\mathcal{S} \leftarrow \emptyset$
\FOR{each update time $t$}
  \STATE Form a candidate from the latest $T_{s}$ of $(X,u)$ and admit it if it improves the minimum singular value of the stacked terminal sensitivities
  \STATE Roll out all segments in $\mathcal{S}$ with current $\hat{\theta}$; integrate sensitivities $S_j$ and accumulate $\Xi_{\mathcal{S}},G_{\mathcal{S}}$ by \cref{eq:stack_defs}
  \STATE Update $\Gamma$ by \cref{eq:Gamma_update}, replacing $\Psi$ with $G_{\mathcal{S}}$
  \STATE Update $\hat{\theta}$ by $\dot{\hat{\theta}}=\mathrm{proj}_{\Gamma}(\hat{\theta},\alpha\Xi_{\mathcal{S}}-k_{\sigma}\hat{\theta})$
\ENDFOR
\end{algorithmic}
\end{algorithm}

NODE-CL uses short segments to control sensitivity growth and admits segments that add complementary directions within a fixed memory budget. Before the projected update, $G_{\mathcal{S}}$ rescales those directions. Segment length therefore determines a bias--conditioning tradeoff, while the admission rule remains a data-selection heuristic rather than an excitation guarantee. The Gauss--Newton rescaling distinguishes NODE-CL from replaying the same segments with a generic optimizer. The following theorem certifies NODE-CL.

\begin{maintheorem}
\label{thm:stack}
Let \cref{assm:bounded_trajectory,assm:f_smooth} hold, let \cref{eq:containment_cond} hold at $T_{s}$, and activate NODE-CL at time $t_s$ with a nonempty segment set, $\hat{\theta}(t_s) \in \Theta_{\epsilon}$, and $\underline{\lambda}I_{p} \preceq \Gamma(t_s) \preceq \bar{\lambda}I_{p}$. Suppose $\alpha > \tfrac{1}{2}$ and
\begin{equation}
k_{\sigma} > 2\alpha(b_{1}\bar{\tilde{\theta}} + b_{2}) + g_{2}\bar{\tilde{\theta}}^{2} + g_{1}\bar{\tilde{\theta}} + g_{0}.
\label{eq:stack_gain_condition}
\end{equation}
Then $\underline{\lambda}I_p \preceq \Gamma(t) \preceq \bar{\lambda} I_p$ for all $t \geq t_s$. For any $\varphi_s\geq 0$ such that $H_{\mathcal{S}}(t)\succeq\varphi_s I_p$ for all $t\geq t_s$ (always valid with $\varphi_s=0$), let
\begin{equation}
k_{1s} \triangleq \tfrac{k_{\sigma}}{2} - \alpha(b_{1}\bar{\tilde{\theta}} + b_{2}) - \tfrac{1}{2}(g_{2}\bar{\tilde{\theta}}^{2} + g_{1}\bar{\tilde{\theta}} + g_{0}),
\quad k_s(\varphi_s) \triangleq k_{1s}+(\alpha-\tfrac{1}{2})\varphi_s,
\quad \lambda_s(\varphi_s) \triangleq 2k_s(\varphi_s)\underline{\lambda}.
\label{eq:stack_rate}
\end{equation}
The parameter error satisfies
\begin{equation}
\left\Vert \tilde{\theta}(t) \right\Vert \leq \sqrt{\frac{\bar{\lambda}}{\underline{\lambda}}\left\Vert \tilde{\theta}(t_s) \right\Vert^2 e^{-\lambda_s(\varphi_s)(t-t_s)} + \frac{\bar{\lambda}c_a}{\underline{\lambda}k_s(\varphi_s)}},
\qquad
\limsup_{t\to\infty}\left\Vert \tilde{\theta}(t) \right\Vert \leq \sqrt{\frac{\bar{\lambda}c_a}{\underline{\lambda}k_s(\varphi_s)}}.
\label{eq:stack_bound}
\end{equation}
\end{maintheorem}

The proof is given in \cref{app:stack_proof}. Thus excitation increases $k_s(\varphi_s)$, improving both the exponential rate and the certified ultimate radius; without excitation, set $\varphi_s=0$. The ratio $\bar{\lambda}/\underline{\lambda}$ in \cref{eq:stack_bound} arises only from converting the Lyapunov metric to the Euclidean norm. By the same bound on $G_{\mathcal{S}} - H_{\mathcal{S}}$, the computable condition $G_{\mathcal{S}} \succeq \varphi I_{p}$ implies $H_{\mathcal{S}} \succeq (\varphi - g_{2}\bar{\tilde{\theta}}^{2} - g_{1}\bar{\tilde{\theta}} - g_{0}) I_{p}$, so the enhanced rate holds with that $\varphi_{s}$ whenever it is positive. Segments are admitted by the singular-value maximizing rule of \citet{Chowdhary2011} on the stacked terminal sensitivities $S_{j}(t_{j}+T_{s})$, a heuristic that does not itself guarantee excitation. The complete procedure is given in \cref{alg:stack}; pseudocode for the drift-window and adjoint-window methods is in \cref{app:algorithms}. \cref{sec:simulations} evaluates the three methods.

\section{Simulations}
\label{sec:simulations}

The experiments test whether the observer-free trajectory pathway is robust to velocity noise while staying competitive on clean data; two diagnostic systems then check individual theorem conditions. \cref{sec:bench} reports the benchmark and \cref{sec:diag} the diagnostics.

\subsection{Benchmark suite}
\label{sec:bench}

\textbf{Methods.} NODE-CL is the proposed stored-segment method evaluated here. Baselines are single-step adaptation, CL-LS (the identification form of the least-squares concurrent-learning law for DNNs of \citet{Hart2025}, the stored-set extension of \cref{thm:drift} in \cref{app:stored_proof}), CL (the same DNN point stack with a constant gain, i.e., the classical stack of \citet{Chowdhary2011} applied to the DNN Jacobian, a simplification of \citet{Hart2025} rather than a separate prior method), and NODE-replay (Adam \citep{Kingma2015} applied to exactly the NODE-CL segments). The drift- and adjoint-window methods are evaluated in the diagnostics.

\textbf{Protocol.} Four DeepMind Control Suite domains \citep{Tassa2018}, pendulum, cartpole, acrobot, and reacher, are used unmodified, and the full acceleration map $\ddot{x} = \Phi(z, \theta)$ with $z = [\phi(x)^{\top}, \dot{x}^{\top}, u^{\top}]^{\top}$ is identified, where $\phi$ maps unlimited hinge angles to $(\cos, \sin)$; no input matrix is assumed known, and \cref{sec:adjoint} applies with $F = [\chi_{v}^{\top}, \Phi(z, \theta)^{\top}]^{\top}$. Each run streams $60\,$s of $1\,$kHz measurements under a multisine excitation with a weak stabilizing baseline; five evaluation seeds vary excitation, initialization, and noise; a sixth, tuning-only seed is excluded. Evaluation is open-loop prediction from eight held-out initial conditions on a held-out record: root-mean-square state error normalized by the state standard deviation at $1\,$s and $4\,$s horizons, angles wrapped. All methods use the same network (one hidden layer, $73$ to $162$ parameters) and, per seed and noise level, the same measurement stream; CL and CL-LS share a point stack, NODE-replay and NODE-CL a segment budget and admission rule, the observer-based methods the observer, and the projected laws the projection. Gains were selected on the tuning seed and held fixed;\footnote{As is standard for Lyapunov-based nonlinear and adaptive designs, the gain conditions in \cref{thm:drift,thm:adjoint_uub,thm:stack} are sufficient, not necessary, and conservative (worst-case Lipschitz and Taylor constants); the gains were tuned for performance and, as \cref{fig:ablation}(c) and \cref{app:sim_details} show, the sampled implementations remain stable well below the certified thresholds.} a six-link swimmer ($p = 356$) that no method learns at this budget is reported in \cref{app:sim_details}.

\begin{table}[t]
\caption{Benchmark suite (five evaluation seeds, tuning seed excluded, $60\,$s): held-out $1\,$s prediction error normalized by the state standard deviation, median over seeds, clean (c) and with velocity noise $\sigma_{v} = 3\times10^{-3}$ (n); in parentheses, seeds on which it is best; best median in bold. Initial-model errors: $2.16$, $2.08$, $1.52$, $1.71$. CL-LS is the identification form of the DNN concurrent-learning law of \citet{Hart2025} (stored-set extension of \cref{thm:drift}, \cref{app:stored_proof}); CL is its constant-gain simplification; NODE-CL is the proposed method in \cref{thm:stack}.}
\label{tab:bench}
\centering
\scriptsize
\setlength{\tabcolsep}{3.25pt}
\begin{tabular}{lcccccccc}
\toprule
 & \multicolumn{2}{c}{pendulum ($p=73$)} & \multicolumn{2}{c}{cartpole ($p=146$)} & \multicolumn{2}{c}{acrobot ($p=162$)} & \multicolumn{2}{c}{reacher ($p=162$)} \\
\cmidrule(lr){2-3}\cmidrule(lr){4-5}\cmidrule(lr){6-7}\cmidrule(lr){8-9}
Method & c & n & c & n & c & n & c & n \\
\midrule
single-step & $1.643$\,(0) & $1.650$\,(0) & $3.287$\,(0) & $3.145$\,(0) & $1.621$\,(1) & $\mathbf{1.232}$\,(2) & $1.509$\,(0) & $1.516$\,(0) \\
CL (constant gain) & $0.315$\,(0) & $0.609$\,(0) & $0.623$\,(0) & $1.363$\,(0) & $3.615$\,(0) & $4.385$\,(0) & $1.007$\,(0) & $0.976$\,(0) \\
CL-LS \citep{Hart2025} & $0.195$\,(0) & $1.642$\,(0) & $\mathbf{0.251}$\,(4) & $1.984$\,(0) & $1.409$\,(0) & $2.015$\,(0) & $\mathbf{0.064}$\,(2) & $0.215$\,(0) \\
NODE-replay (control) & $0.183$\,(0) & $0.141$\,(3) & $0.311$\,(0) & $0.741$\,(1) & $\mathbf{1.002}$\,(2) & $2.129$\,(1) & $0.077$\,(2) & $0.088$\,(2) \\
NODE-CL (proposed) & $\mathbf{0.035}$\,(5) & $\mathbf{0.038}$\,(2) & $0.381$\,(1) & $\mathbf{0.302}$\,(4) & $1.545$\,(2) & $1.279$\,(2) & $0.069$\,(1) & $\mathbf{0.075}$\,(3) \\
\bottomrule
\end{tabular}
\end{table}

\begin{figure}[t]
\centering
\includegraphics[width=0.82\textwidth]{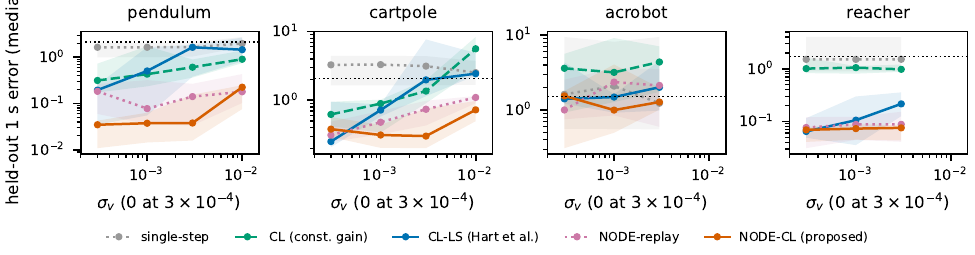}
\caption{Held-out $1\,$s error against velocity noise per domain (five seeds, median and range; dotted, initial model). The $10^{-2}$ level was not run for the acrobot and reacher.}
\label{fig:bench_noise}
\end{figure}

\textbf{Results (\cref{tab:bench}, \cref{fig:bench_noise}).} With clean measurements, NODE-CL is best on every pendulum seed ($0.035$ against $0.18$ to $0.32$) but not on the other domains. On the cartpole and reacher, the three stored-set methods lie within a factor of $1.6$. Under velocity noise of $3 \times 10^{-3}$, CL-LS degrades on these three domains by a factor of $3.3$ to $8.4$ relative to its clean error, and CL degrades by up to $2.2$; $k_{f} = 200$ was tuned on clean data rather than per noise level. NODE-CL has the lowest median on all three domains and wins nine of fifteen seeds, the expected advantage of trajectory residuals when derivative estimates are corrupted by noise. The advantage is not uniform: on the pendulum NODE-replay wins three seeds at $3\times10^{-3}$ despite a higher median, and at $10^{-2}$ NODE-CL is ahead of NODE-replay on the cartpole ($0.73$ against $1.10$) but behind on the pendulum ($0.23$ against $0.18$). On the acrobot, NODE-CL improves on the initial model at $10^{-3}$ and $3\times10^{-3}$ ($1.00$ and $1.28$ against $1.52$) but not with clean data ($1.55$), whereas NODE-replay does so only with clean data ($1.00$). \cref{sec:diag} tests individual conditions of the analysis on two diagnostic systems.

\subsection{Diagnostic studies and ablations}
\label{sec:diag}

Two diagnostic systems, a teacher network (S1, $p=25$) with known $\theta^{*}$ and a two-link manipulator (S2, $p=86$), test distinct predictions of the analysis. \textbf{Excited subspace:} in S1, the component of the parameter error in the excited subspace contracts to $0.27$ of its initial value while the ratio for the full parameter error remains near $0.80$, consistent with the excitation-free functional bound of \cref{app:drift_proof}. \textbf{Horizon:} the adjoint-window error increases about $3.7$-fold from $T = 0.8$ to $1.6\,$s, consistent with the horizon dependence in \cref{eq:etaa_bound}. \textbf{Sufficient gains:} sampled implementations remain stable below the conservative thresholds on $\alpha$ and leakage. Full curves, seed ranges, and condition checks are in \cref{app:sim_details}. \cref{sec:conclusion} concludes.

\section{Conclusion}
\label{sec:conclusion}

We showed that online Neural ODE training and concurrent learning fit in one Lyapunov-certified framework for adaptive identification: drift and trajectory residuals both yield a gradient of the form regressor times parameter error plus a bounded perturbation, so one projected least-squares update is certified for windowed and stored data, and its stored-segment instance, NODE-CL, brings a Gauss--Newton gain to concurrent learning without state-derivative estimates. On the pendulum, cartpole, and reacher, NODE-CL is best at $\sigma_v = 3 \times 10^{-3}$, where the observer-based CL-LS degrades sharply; with clean data it is best on the pendulum and within a factor of $1.6$ of the best method on the cartpole and reacher, and on the acrobot the advantage is not uniform.

\subsection*{Reproducibility statement}

All assumptions are stated in \cref{sec:problem}, and complete proofs of every lemma, theorem, and corollary are given in \cref{app:observer_proof,app:gamma_proof,app:proj_proof,app:drift_decomp_proof,app:drift_proof,app:functional_proof,app:adjoint_proof,app:containment_proof,app:stack_proof,app:adjoint_uub_proof}. NODE-CL pseudocode is given in \cref{alg:stack}, the proposed drift-window and adjoint-window algorithms in \cref{app:algorithms}, and every gain, window length, sampling rate, architecture, excitation signal, and evaluation metric used in the simulations is listed in \cref{sec:simulations} and \cref{app:sim_details}. The supplementary material contains the complete simulation code, the resumable experiment drivers for the main comparison, the noise sweep, and the ablations, the analysis scripts that generate the tables and figures from the run records, the saved summary files from which the tables and \cref{fig:bench_noise} are reproduced, and the numerical checks described in \cref{app:sim_details}; the raw run records are omitted from the archive for size and are regenerated by the drivers. The simulations are deterministic given the seed.

\subsection*{AI use statement}

In this work, generative AI tools were used for implementing the simulation code, for interpreting simulation results, and for drafting and editing the text and formatting references. The two illustrative icons in \cref{fig:overview} (a sensor module and a memory stack) are AI-generated conceptual images; no reported experimental data or numerical results were produced by generative AI, and the research questions were determined by the author(s). Selected benchmark derivatives and numerical updates were checked against the mathematical formulation, including finite-difference checks of the trajectory gradient and Gauss--Newton matrix and numerical checks of the projection inequality, the gain matrix bounds, the observer bound, and the drift gradient decomposition on simulated data; all AI-assisted text was reviewed. Responsibility for the final content of this work, including text, claims, and artifacts produced with the aid of generative AI, rests with the author(s).

\bibliography{references}

@inproceedings{Chen2018,
  author    = {Chen, Ricky T. Q. and Rubanova, Yulia and Bettencourt, Jesse and Duvenaud, David},
  title     = {Neural Ordinary Differential Equations},
  booktitle = {Advances in Neural Information Processing Systems},
  volume    = {31},
  year      = {2018}
}

@article{SlotineLi1989,
  author  = {Slotine, Jean-Jacques E. and Li, Weiping},
  title   = {Composite Adaptive Control of Robot Manipulators},
  journal = {Automatica},
  volume  = {25},
  number  = {4},
  pages   = {509--519},
  year    = {1989}
}

@article{Chowdhary2011,
  author  = {Chowdhary, Girish and Johnson, Eric},
  title   = {Theory and Flight-Test Validation of a Concurrent-Learning Adaptive Controller},
  journal = {Journal of Guidance, Control, and Dynamics},
  volume  = {34},
  number  = {2},
  pages   = {592--607},
  year    = {2011}
}

@article{Patil2022,
  author  = {Patil, Omkar Sudhir and Le, Duc M. and Greene, Max L. and Dixon, Warren E.},
  title   = {Lyapunov-Derived Control and Adaptive Update Laws for Inner and Outer Layer Weights of a Deep Neural Network},
  journal = {IEEE Control Systems Letters},
  volume  = {6},
  pages   = {1855--1860},
  year    = {2022}
}

@book{Khalil2002,
  author    = {Khalil, Hassan K.},
  title     = {Nonlinear Systems},
  edition   = {3rd},
  publisher = {Prentice Hall},
  year      = {2002}
}

@book{Lavretsky2013,
  author    = {Lavretsky, Eugene and Wise, Kevin A.},
  title     = {Robust and Adaptive Control with Aerospace Applications},
  publisher = {Springer},
  year      = {2013}
}

@article{Cai2006,
  author  = {Cai, Zhijun and de Queiroz, Marcio S. and Dawson, Darren M.},
  title   = {A Sufficiently Smooth Projection Operator},
  journal = {IEEE Transactions on Automatic Control},
  volume  = {51},
  number  = {1},
  pages   = {135--139},
  year    = {2006}
}

@book{Lewis1999,
  author    = {Lewis, Frank L. and Jagannathan, Sarangapani and Yesildirek, Aydin},
  title     = {Neural Network Control of Robot Manipulators and Nonlinear Systems},
  publisher = {Taylor \& Francis},
  year      = {1999}
}

@inproceedings{Kingma2015,
  author    = {Kingma, Diederik P. and Ba, Jimmy},
  title     = {Adam: A Method for Stochastic Optimization},
  booktitle = {International Conference on Learning Representations},
  year      = {2015}
}

@article{Hart2025,
  author  = {Hart, Rebecca G. and Patil, Omkar Sudhir and Bell, Zachary I. and Dixon, Warren E.},
  title   = {Concurrent Learning for System Identification and Control Using {L}yapunov-Based Deep Neural Networks},
  journal = {IEEE Control Systems Letters},
  volume  = {9},
  pages   = {2957--2962},
  year    = {2025}
}

@article{Hart2025arxiv,
  author  = {Hart, Rebecca G. and Patil, Omkar Sudhir and Bell, Zachary I. and Dixon, Warren E.},
  title   = {System Identification and Control Using {L}yapunov-Based Deep Neural Networks without Persistent Excitation: A Concurrent Learning Approach},
  journal = {arXiv preprint arXiv:2505.10678},
  year    = {2025}
}

@article{Makumi2025,
  author  = {Makumi, Wanjiku and Patil, Omkar Sudhir and Dixon, Warren E.},
  title   = {Lyapunov-Based Adaptive Deep System Identification for Approximate Dynamic Programming},
  journal = {Automatica},
  volume  = {180},
  pages   = {112462},
  year    = {2025}
}

@article{Tassa2018,
  author  = {Tassa, Yuval and Doron, Yotam and Muldal, Alistair and Erez, Tom and Li, Yazhe and de Las Casas, Diego and Budden, David and Abdolmaleki, Abbas and Merel, Josh and Lefrancq, Andrew and Lillicrap, Timothy and Riedmiller, Martin},
  title   = {{DeepMind} Control Suite},
  journal = {arXiv preprint arXiv:1801.00690},
  year    = {2018}
}

@article{Turan2022,
  author  = {Turan, Evren Mert and J{\"a}schke, Johannes},
  title   = {Multiple Shooting for Training Neural Differential Equations on Time Series},
  journal = {IEEE Control Systems Letters},
  volume  = {6},
  pages   = {1897--1902},
  year    = {2022}
}

@inproceedings{Rubanova2019,
  author    = {Rubanova, Yulia and Chen, Ricky T. Q. and Duvenaud, David},
  title     = {Latent Ordinary Differential Equations for Irregularly-Sampled Time Series},
  booktitle = {Advances in Neural Information Processing Systems},
  volume    = {32},
  year      = {2019}
}

@phdthesis{Kidger2022,
  author = {Kidger, Patrick},
  title  = {On Neural Differential Equations},
  school = {University of Oxford},
  year   = {2022}
}

@article{Brunton2016,
  author  = {Brunton, Steven L. and Proctor, Joshua L. and Kutz, J. Nathan},
  title   = {Discovering Governing Equations from Data by Sparse Identification of Nonlinear Dynamical Systems},
  journal = {Proceedings of the National Academy of Sciences},
  volume  = {113},
  number  = {15},
  pages   = {3932--3937},
  year    = {2016}
}

@article{Sanner1992,
  author  = {Sanner, Robert M. and Slotine, Jean-Jacques E.},
  title   = {Gaussian Networks for Direct Adaptive Control},
  journal = {IEEE Transactions on Neural Networks},
  volume  = {3},
  number  = {6},
  pages   = {837--863},
  year    = {1992}
}

@inproceedings{Joshi2019,
  author    = {Joshi, Girish and Chowdhary, Girish},
  title     = {Deep Model Reference Adaptive Control},
  booktitle = {IEEE Conference on Decision and Control},
  pages     = {4601--4608},
  year      = {2019}
}

@article{Sun2022,
  author  = {Sun, Runhan and Greene, Max L. and Le, Duc M. and Bell, Zachary I. and Chowdhary, Girish and Dixon, Warren E.},
  title   = {Lyapunov-Based Real-Time and Iterative Adjustment of Deep Neural Networks},
  journal = {IEEE Control Systems Letters},
  volume  = {6},
  pages   = {193--198},
  year    = {2022}
}

@article{Kamalapurkar2017,
  author  = {Kamalapurkar, Rushikesh and Reish, Benjamin and Chowdhary, Girish and Dixon, Warren E.},
  title   = {Concurrent Learning for Parameter Estimation Using Dynamic State-Derivative Estimators},
  journal = {IEEE Transactions on Automatic Control},
  volume  = {62},
  number  = {7},
  pages   = {3594--3601},
  year    = {2017}
}

@article{Boffi2021,
  author  = {Boffi, Nicholas M. and Slotine, Jean-Jacques E.},
  title   = {Implicit Regularization and Momentum Algorithms in Nonlinearly Parameterized Adaptive Control and Prediction},
  journal = {Neural Computation},
  volume  = {33},
  number  = {3},
  pages   = {590--673},
  year    = {2021}
}

@inproceedings{BoffiTu2021,
  author    = {Boffi, Nicholas M. and Tu, Stephen and Slotine, Jean-Jacques E.},
  title     = {Regret Bounds for Adaptive Nonlinear Control},
  booktitle = {Proceedings of the 3rd Conference on Learning for Dynamics and Control (L4DC)},
  series    = {Proceedings of Machine Learning Research},
  volume    = {144},
  pages     = {471--483},
  publisher = {PMLR},
  year      = {2021}
}

@book{Ljung1999,
  author    = {Ljung, Lennart},
  title     = {System Identification: Theory for the User},
  publisher = {Prentice Hall},
  edition   = {2nd},
  year      = {1999}
}

@inproceedings{SinghalWu1989,
  author    = {Singhal, Sharad and Wu, Lance},
  title     = {Training Multilayer Perceptrons with the Extended {K}alman Algorithm},
  booktitle = {Advances in Neural Information Processing Systems},
  volume    = {1},
  year      = {1989}
}

@article{Puskorius1994,
  author  = {Puskorius, Gintaras V. and Feldkamp, Lee A.},
  title   = {Neurocontrol of Nonlinear Dynamical Systems with {K}alman Filter Trained Recurrent Networks},
  journal = {IEEE Transactions on Neural Networks},
  volume  = {5},
  number  = {2},
  pages   = {279--297},
  year    = {1994}
}

@book{Haykin2001,
  editor    = {Haykin, Simon},
  title     = {Kalman Filtering and Neural Networks},
  publisher = {Wiley},
  year      = {2001}
}

@inproceedings{Kolter2019,
  author    = {Kolter, J. Zico and Manek, Gaurav},
  title     = {Learning Stable Deep Dynamics Models},
  booktitle = {Advances in Neural Information Processing Systems},
  volume    = {32},
  year      = {2019}
}

@inproceedings{Dupont2019,
  author    = {Dupont, Emilien and Doucet, Arnaud and Teh, Yee Whye},
  title     = {Augmented Neural {ODEs}},
  booktitle = {Advances in Neural Information Processing Systems},
  volume    = {32},
  year      = {2019}
}

@inproceedings{Massaroli2020,
  author    = {Massaroli, Stefano and Poli, Michael and Park, Jinkyoo and Yamashita, Atsushi and Asama, Hajime},
  title     = {Dissecting Neural {ODEs}},
  booktitle = {Advances in Neural Information Processing Systems},
  volume    = {33},
  pages     = {3952--3963},
  year      = {2020}
}

@inproceedings{Kidger2020,
  author    = {Kidger, Patrick and Morrill, James and Foster, James and Lyons, Terry},
  title     = {Neural Controlled Differential Equations for Irregular Time Series},
  booktitle = {Advances in Neural Information Processing Systems},
  volume    = {33},
  pages     = {6696--6707},
  year      = {2020}
}

@inproceedings{Finlay2020,
  author    = {Finlay, Chris and Jacobsen, J{\"o}rn-Henrik and Nurbekyan, Levon and Oberman, Adam M.},
  title     = {How to Train Your Neural {ODE}: The World of {J}acobian and Kinetic Regularization},
  booktitle = {Proceedings of the 37th International Conference on Machine Learning},
  series    = {Proceedings of Machine Learning Research},
  volume    = {119},
  pages     = {3154--3164},
  publisher = {PMLR},
  year      = {2020}
}

@inproceedings{Zhuang2020,
  author    = {Zhuang, Juntang and Dvornek, Nicha and Li, Xiaoxiao and Tatikonda, Sekhar and Papademetris, Xenophon and Duncan, James},
  title     = {Adaptive Checkpoint Adjoint Method for Gradient Estimation in Neural {ODE}},
  booktitle = {Proceedings of the 37th International Conference on Machine Learning},
  series    = {Proceedings of Machine Learning Research},
  volume    = {119},
  pages     = {11639--11649},
  publisher = {PMLR},
  year      = {2020}
}

@inproceedings{Massaroli2021,
  author    = {Massaroli, Stefano and Poli, Michael and Sonoda, Sho and Suzuki, Taiji and Park, Jinkyoo and Yamashita, Atsushi and Asama, Hajime},
  title     = {Differentiable Multiple Shooting Layers},
  booktitle = {Advances in Neural Information Processing Systems},
  volume    = {34},
  pages     = {16532--16544},
  year      = {2021}
}

@inproceedings{Kang2021,
  author    = {Kang, Qiyu and Song, Yang and Ding, Qinxu and Tay, Wee Peng},
  title     = {Stable Neural {ODE} with {L}yapunov-Stable Equilibrium Points for Defending Against Adversarial Attacks},
  booktitle = {Advances in Neural Information Processing Systems},
  volume    = {34},
  pages     = {14925--14937},
  year      = {2021}
}

@inproceedings{Rodriguez2022,
  author    = {Rodriguez, Ivan Dario Jimenez and Ames, Aaron D. and Yue, Yisong},
  title     = {{LyaNet}: A {L}yapunov Framework for Training Neural {ODEs}},
  booktitle = {Proceedings of the 39th International Conference on Machine Learning},
  series    = {Proceedings of Machine Learning Research},
  volume    = {162},
  pages     = {18687--18703},
  publisher = {PMLR},
  year      = {2022}
}

@inproceedings{White2023,
  author    = {White, Alistair and Kilbertus, Niki and Gelbrecht, Maximilian and Boers, Niklas},
  title     = {Stabilized Neural Differential Equations for Learning Dynamics with Explicit Constraints},
  booktitle = {Advances in Neural Information Processing Systems},
  volume    = {36},
  pages     = {12929--12950},
  year      = {2023}
}

@inproceedings{Sochopoulos2024,
  author    = {Sochopoulos, Andreas and Gienger, Michael and Vijayakumar, Sethu},
  title     = {Learning Deep Dynamical Systems Using Stable Neural {ODEs}},
  booktitle = {IEEE/RSJ International Conference on Intelligent Robots and Systems},
  pages     = {11163--11170},
  year      = {2024}
}

@article{Luo2025,
  author  = {Luo, Chaoyang and Zou, Yan and Li, Wanying and Huang, Nanjing},
  title   = {{FxTS-Net}: Fixed-Time Stable Learning Framework for Neural {ODEs}},
  journal = {Neural Networks},
  volume  = {185},
  pages   = {107219},
  year    = {2025}
}

@inproceedings{Richards2021,
  author    = {Richards, Spencer M. and Azizan, Navid and Slotine, Jean-Jacques and Pavone, Marco},
  title     = {Adaptive-Control-Oriented Meta-Learning for Nonlinear Systems},
  booktitle = {Proceedings of Robotics: Science and Systems},
  year      = {2021}
}

@inproceedings{Shi2021,
  author    = {Shi, Guanya and Azizzadenesheli, Kamyar and O'Connell, Michael and Chung, Soon-Jo and Yue, Yisong},
  title     = {Meta-Adaptive Nonlinear Control: Theory and Algorithms},
  booktitle = {Advances in Neural Information Processing Systems},
  volume    = {34},
  pages     = {10013--10025},
  year      = {2021}
}

@article{OConnell2022,
  author  = {O'Connell, Michael and Shi, Guanya and Shi, Xichen and Azizzadenesheli, Kamyar and Anandkumar, Anima and Yue, Yisong and Chung, Soon-Jo},
  title   = {Neural-{F}ly Enables Rapid Learning for Agile Flight in Strong Winds},
  journal = {Science Robotics},
  volume  = {7},
  number  = {66},
  pages   = {eabm6597},
  year    = {2022}
}
\bibliographystyle{iclr2027_conference}

\clearpage
\appendix

\section{Drift Observer and Error Bound}
\label{app:observer_proof}
\label{sec:observer}

The experiments instantiate the drift estimate with the following state-derivative observer. Let $\hat{v} \in \mathbb{R}^{n}$ denote an auxiliary filter state and $\tilde{v} \triangleq \dot{x}-\hat{v}$ its measurable error. For gains $\alpha_{o},k_{f}>0$, define
\begin{align}
\dot{\hat{v}} &= \alpha_{o}\tilde{v}+\hat{f}+g(X)u, \label{eq:vhat_dot} \\
\hat{f}(t) &= \hat{f}(t_{0})+k_{f}\left(\tilde{v}(t)-\tilde{v}(t_{0})\right)+\int_{t_{0}}^{t}\left(k_{f}\alpha_{o}+1\right)\tilde{v}(\tau)\,d\tau. \label{eq:fhat_int}
\end{align}
Differentiating \cref{eq:fhat_int} and using \cref{eq:system,eq:vhat_dot}, the errors $\tilde{v}$ and $\tilde{f}=f(X)-\hat{f}$ satisfy $\dot{\tilde{v}}=\tilde{f}-\alpha_{o}\tilde{v}$ and $\dot{\tilde{f}}=\dot{f}-k_{f}\tilde{f}-\tilde{v}$. Let $z_{o}\triangleq[\tilde{v}^{\top}\;\tilde{f}^{\top}]^{\top}$.

\begin{lemma}
\label{lem:observer}
Let \cref{assm:bounded_trajectory,assm:f_smooth} hold. Then, for all $t\geq t_{0}$,
\begin{equation}
\label{eq:observer_bound}
\left\Vert z_{o}(t) \right\Vert \leq \left\Vert z_{o}(t_{0}) \right\Vert e^{-k_{o}(t-t_{0})}+\bar{z}_{o},
\end{equation}
where $k_{o}\triangleq\min\{\alpha_{o},k_{f}/2\}$, $\bar{z}_{o}\triangleq d_{f}/\sqrt{2k_{f}k_{o}}$, and $d_{f}\triangleq c_{fx}\bar{X}+c_{fv}(\bar{f}+\bar{g}\bar{u})$.
\end{lemma}

\textit{Proof.} Consider the observer Lyapunov function $V_{o} = \tfrac{1}{2}\tilde{v}^{\top}\tilde{v} + \tfrac{1}{2}\tilde{f}^{\top}\tilde{f}$, which satisfies $V_{o} = \tfrac{1}{2}\left\Vert z_{o} \right\Vert^{2}$. Taking the time derivative along the observer error dynamics yields
\begin{equation}
\dot{V}_{o} = \tilde{v}^{\top}(\tilde{f} - \alpha_{o}\tilde{v}) + \tilde{f}^{\top}(\dot{f} - k_{f}\tilde{f} - \tilde{v})
= -\alpha_{o}\left\Vert \tilde{v} \right\Vert^{2} - k_{f}\left\Vert \tilde{f} \right\Vert^{2} + \tilde{f}^{\top}\dot{f},
\end{equation}
where the cross terms cancel. Under \cref{assm:bounded_trajectory,assm:f_smooth}, the drift derivative along the trajectory satisfies
\begin{equation}
\left\Vert \dot{f} \right\Vert = \left\Vert \tfrac{\partial f}{\partial x}\dot{x} + \tfrac{\partial f}{\partial \dot{x}}\left(f + g u\right) \right\Vert \leq c_{fx}\bar{X} + c_{fv}\left(\bar{f} + \bar{g}\bar{u}\right) = d_{f}.
\end{equation}
Using Young's inequality, $\tilde{f}^{\top}\dot{f} \leq \tfrac{k_{f}}{2}\left\Vert \tilde{f} \right\Vert^{2} + \tfrac{d_{f}^{2}}{2k_{f}}$, and therefore
\begin{equation}
\dot{V}_{o} \leq -k_{o}\left\Vert z_{o} \right\Vert^{2} + c_{o} = -2k_{o}V_{o} + c_{o},
\end{equation}
where $k_{o} = \min\{\alpha_{o}, k_{f}/2\}$ and $c_{o} = d_{f}^{2}/(2k_{f})$. The comparison lemma \citep{Khalil2002} yields $V_{o}(t) \leq V_{o}(t_{0})e^{-2k_{o}(t-t_{0})} + \tfrac{c_{o}}{2k_{o}}$, and using $\left\Vert z_{o} \right\Vert = \sqrt{2V_{o}}$ together with $\sqrt{a+b} \leq \sqrt{a} + \sqrt{b}$ gives \cref{eq:observer_bound} with $\bar{z}_{o} = \sqrt{c_{o}/k_{o}} = d_{f}/\sqrt{2k_{f}k_{o}}$. \hfill $\blacksquare$

\section{Pseudocode for the Proposed Drift-Window and Adjoint-Window Methods}
\label{app:algorithms}

The two windowed laws of \cref{sec:drift} are given in pseudocode here, in the form implemented in the experiments.

\setcounter{algorithm}{0}
\begin{algorithm}[H]
\caption{Proposed Drift-Window Identification}
\label{alg:drift}
\begin{algorithmic}[1]
\REQUIRE $T$, gains $\alpha_{o}, k_{f}, \alpha, k_{\sigma}, \beta_{0}, \bar{\lambda}, \underline{\lambda}$
\STATE Initialize $\hat{\theta}(t_{0}) \in \Theta_{\epsilon}$, $\Gamma(t_{0}) = \Gamma_{0}$, $\hat{v}(t_{0}) = \dot{x}(t_{0})$, $\hat{f}(t_{0}) = 0$, empty circular buffer
\FOR{each time step $t$}
  \STATE Measure $x, \dot{x}$; record $u$
  \STATE Update observer \cref{eq:vhat_dot}--\cref{eq:fhat_int}; push $(X(t), \hat{f}(t))$ to buffer
  \IF{$t \geq t_{0} + T$}
    \STATE $E(\tau) \leftarrow \hat{f}(\tau) - \Phi(X(\tau), \hat{\theta})$ over the buffer
    \STATE $\Xi_{f} \leftarrow$ quadrature of $\Phi'(X(\tau),\hat{\theta})^{\top} E(\tau)$
    \STATE $\Psi \leftarrow$ quadrature of $\Phi'(X(\tau),\hat{\theta})^{\top} \Phi'(X(\tau),\hat{\theta})$
  \ELSE
    \STATE $\Xi_{f} \leftarrow 0_{p}$; $\Psi \leftarrow 0_{p \times p}$
  \ENDIF
  \STATE Update $\Gamma$ by \cref{eq:Gamma_update} with gate $\varsigma_{\Gamma}$
  \STATE Update $\hat{\theta}$ by \cref{eq:adapt_law}
\ENDFOR
\end{algorithmic}
\end{algorithm}

\begin{algorithm}[H]
\caption{Proposed Adjoint-Window Identification}
\label{alg:adjoint}
\begin{algorithmic}[1]
\REQUIRE $T$, gains $\alpha, k_{\sigma}$, constant $\Gamma_{0}$
\STATE Initialize $\hat{\theta}(t_{0}) \in \Theta_{\epsilon}$, empty circular buffer
\FOR{each time step $t$}
  \STATE Measure $x, \dot{x}$; record $u$; push $(X(t), u(t))$ to buffer
  \IF{$t \geq t_{0} + T$}
    \STATE Forward: integrate \cref{eq:virtual_system} from $X(t-T)$
    \STATE Backward: integrate the adjoint equation in \cref{eq:adjoint_gradient} from $\nu(t) = 0$
    \STATE $\Xi_{a} \leftarrow -$ quadrature of $\Phi'(\hat{X}_{p},\hat{\theta})^{\top} \nu_{v}$
  \ELSE
    \STATE $\Xi_{a} \leftarrow 0_{p}$
  \ENDIF
  \STATE Update $\hat{\theta}$ by \cref{eq:adapt_law_a}
\ENDFOR
\end{algorithmic}
\end{algorithm}

\section{Gain-Matrix Bounds}
\label{app:gamma_proof}

This appendix proves that the gated forgetting update keeps the gain matrix inside the certified interval.

\begin{lemma}
\label{lem:gamma}
If $\underline{\lambda} I_{p} \preceq \Gamma_{0} \preceq \bar{\lambda} I_{p}$, then the solution of \cref{eq:Gamma_update} satisfies $\underline{\lambda} I_{p} \preceq \Gamma(t) \preceq \bar{\lambda} I_{p}$ for all $t \geq t_{0}$, and $\tfrac{d}{dt}\Gamma^{-1} = -\beta_{0}\left(\Gamma^{-1} - \bar{\lambda}^{-1}I_{p}\right) + \varsigma_{\Gamma}\Psi$ with $\Gamma^{-1} - \bar{\lambda}^{-1}I_{p} \succeq 0$.
\end{lemma}

The right-hand side of \cref{eq:Gamma_update} is locally Lipschitz in $\Gamma$ (it is polynomial in $\Gamma$, and $\varsigma_{\Gamma}$ is Lipschitz in $\Gamma$ through $\lambda_{\min}$), and $\Psi$ is piecewise continuous, so solutions are unique and absolutely continuous, and Nagumo's theorem applies to the closed convex sets $\mathcal{K}_{u} \triangleq \{\Gamma : \Gamma \preceq \bar{\lambda}I_{p}\}$ and $\mathcal{K}_{\ell} \triangleq \{\Gamma : \Gamma \succeq \underline{\lambda}I_{p}\}$: each set is forward invariant if, at every boundary point, the vector field lies in the tangent cone, which for these sets holds if $v^{\top}\dot{\Gamma}v \leq 0$ (respectively $\geq 0$) for every unit vector $v$ in the extremal eigenspace.

For the upper bound, let $\Gamma \in \partial\mathcal{K}_{u}$, so $\lambda_{\max}(\Gamma) = \bar{\lambda}$, and let $v$ be any unit vector with $\Gamma v = \bar{\lambda}v$. Then $v^{\top}\dot{\Gamma}v = \beta_{0}\bar{\lambda}(1 - \bar{\lambda}/\bar{\lambda}) - \varsigma_{\Gamma}v^{\top}\Gamma\Psi\Gamma v \leq 0$. Since $\Gamma_{0} \in \mathcal{K}_{u}$, forward invariance yields $\Gamma(t) \preceq \bar{\lambda}I_{p}$ for all $t$, hence $\Gamma^{-1} - \bar{\lambda}^{-1}I_{p} \succeq 0$. The expression for $\tfrac{d}{dt}\Gamma^{-1}$ follows from $\tfrac{d}{dt}\Gamma^{-1} = -\Gamma^{-1}\dot{\Gamma}\Gamma^{-1}$.

For the lower bound, let $\Gamma \in \partial\mathcal{K}_{\ell}$, so $\lambda_{\min}(\Gamma) = \underline{\lambda}$, and let $v$ be any unit vector with $\Gamma v = \underline{\lambda}v$. Then $\varsigma_{\Gamma} = \varsigma(\underline{\lambda}) = 0$, so $v^{\top}\dot{\Gamma}v = \beta_{0}\underline{\lambda}(1 - \underline{\lambda}/\bar{\lambda}) > 0$ since $\underline{\lambda} < \bar{\lambda}$. Since $\Gamma_{0} \in \mathcal{K}_{\ell}$, forward invariance yields $\Gamma(t) \succeq \underline{\lambda}I_{p}$ for all $t$. \hfill $\blacksquare$

\section{Gain-Metric Projection}
\label{app:proj_proof}

This appendix defines the gain-metric projection used in every adaptation law and proves its two properties. The projection operator is defined as follows. Let $p_{c}(\theta) \triangleq (\left\Vert \theta \right\Vert^{2} - \bar{\theta}^{2})/(2\epsilon_{\theta}\bar{\theta} + \epsilon_{\theta}^{2})$, so that $p_{c}(\theta) \leq 0$ if and only if $\theta \in \Theta$ and $p_{c}(\theta) = 1$ if and only if $\left\Vert \theta \right\Vert = \bar{\theta} + \epsilon_{\theta}$. For a symmetric positive definite $\Gamma$, the projection operator is defined as
\begin{equation}
\label{eq:proj_def}
\mathrm{proj}_{\Gamma}(\theta, y) \triangleq
\begin{cases}
\Gamma y - p_{c}(\theta)\, \dfrac{\Gamma \nabla p_{c} \nabla p_{c}^{\top} \Gamma\, y}{\nabla p_{c}^{\top} \Gamma \nabla p_{c}}, & \text{if } p_{c}(\theta) > 0 \text{ and } \nabla p_{c}^{\top} \Gamma y > 0, \\[2mm]
\Gamma y, & \text{otherwise},
\end{cases}
\end{equation}
where $\nabla p_{c} = \nabla p_{c}(\theta)$; cf. the smooth projection operators in \citet{Lavretsky2013,Cai2006}. The operator removes only the component of $\Gamma y$ that points outward across the boundary of $\Theta_{\epsilon}$, and it does so in the $\Gamma$ metric, which yields the following properties.

\begin{lemma}
\label{lem:proj}
Let $\Gamma(t)$ be continuous, symmetric, and positive definite for each $t$, and let $\hat{\theta}$ evolve by $\dot{\hat{\theta}} = \mathrm{proj}_{\Gamma}(\hat{\theta}, y(t))$ for a piecewise continuous $y$. If $\hat{\theta}(t_{0}) \in \Theta_{\epsilon}$, then (i) $\hat{\theta}(t) \in \Theta_{\epsilon}$ for all $t \geq t_{0}$, and (ii) for every $\theta^{*} \in \Theta$, \cref{eq:proj_prop} holds.
\end{lemma}

The projection is continuous and locally Lipschitz across both switching surfaces. Together with the smoothness of $\Phi$, this makes the right-hand sides of the proposed adaptation laws locally Lipschitz in $\hat{\theta}$ between data-set updates; hence their Carath\'{e}odory solutions exist and are unique.

For (i), on the boundary $p_{c}(\hat{\theta}) = 1$, the outward normal component of the flow satisfies $\nabla p_{c}^{\top} \dot{\hat{\theta}} = (1 - p_{c}(\hat{\theta}))\, \nabla p_{c}^{\top} \Gamma y = 0$ when the activation condition $\mathcal{A}$ holds, and $\nabla p_{c}^{\top} \dot{\hat{\theta}} = \nabla p_{c}^{\top}\Gamma y \leq 0$ otherwise; hence the sublevel set $\{p_{c} \leq 1\} = \Theta_{\epsilon}$ is forward invariant.

For (ii), when $\mathcal{A}$ does not hold, \cref{eq:proj_prop} holds with equality. When $\mathcal{A}$ holds,
\begin{equation}
\tilde{\theta}^{\top}\Gamma^{-1}\mathrm{proj}_{\Gamma}(\hat{\theta}, y) - \tilde{\theta}^{\top} y = -\, p_{c}(\hat{\theta})\, \frac{\left(\tilde{\theta}^{\top}\nabla p_{c}\right)\left(\nabla p_{c}^{\top}\Gamma y\right)}{\nabla p_{c}^{\top}\Gamma\nabla p_{c}}.
\end{equation}
By convexity of $p_{c}$ and $p_{c}(\theta^{*}) \leq 0 < p_{c}(\hat{\theta})$, it follows that $\tilde{\theta}^{\top}\nabla p_{c} = (\theta^{*} - \hat{\theta})^{\top} \nabla p_{c}(\hat{\theta}) \leq p_{c}(\theta^{*}) - p_{c}(\hat{\theta}) < 0$. Since $p_{c}(\hat{\theta}) > 0$, $\nabla p_{c}^{\top}\Gamma y > 0$, and $\nabla p_{c}^{\top}\Gamma\nabla p_{c} > 0$, the right-hand side is nonnegative, which establishes \cref{eq:proj_prop}. \hfill $\blacksquare$

\section{Proof of the Windowed Gradient Decomposition}
\label{app:drift_decomp_proof}

This appendix proves the windowed gradient decomposition of \cref{lem:drift_decomp}. Using $\hat{f} = f - \tilde{f}$ and \cref{eq:ufap},
\begin{equation}
E(\tau; t) = \Phi(X(\tau), \theta^{*}) + \varepsilon(X(\tau)) - \tilde{f}(\tau) - \Phi(X(\tau), \hat{\theta}(t)).
\end{equation}
Since $\Theta_{\epsilon}$ is convex and contains both $\hat{\theta}(t)$ and $\theta^{*}$, Taylor's theorem with the Lagrange form of the remainder yields
\begin{equation}
\label{eq:taylor}
\Phi(X(\tau), \theta^{*}) - \Phi(X(\tau), \hat{\theta}(t)) = \Phi'(X(\tau), \hat{\theta}(t))\, \tilde{\theta}(t) + R(\tau), \qquad \left\Vert R(\tau) \right\Vert \leq \tfrac{1}{2}c_{\theta\theta}\left\Vert \tilde{\theta}(t) \right\Vert^{2},
\end{equation}
where the remainder bound follows from the bound on $\partial^{2}\Phi/\partial\theta^{2}$. Substituting into \cref{eq:Xif_def} gives \cref{eq:Xif_decomp} with $\Psi$ as in \cref{eq:Psi_def} and
\begin{equation}
\eta_{f}(t) = \int_{t-T}^{t}\Phi'(X(\tau), \hat{\theta}(t))^{\top}\left[R(\tau) + \varepsilon(X(\tau)) - \tilde{f}(\tau)\right] d\tau.
\end{equation}
The matrix $\Psi$ is an integral of Gram matrices and is therefore positive semi-definite with $\left\Vert \Psi \right\Vert \leq T\bar{\upsilon}^{2}$. Using $\left\Vert \Phi' \right\Vert \leq \bar{\upsilon}$, the definition of $\bar{f}_{T}$, and integrating over the window of length $T$ yields \cref{eq:etaf_bound}. \hfill $\blacksquare$

\section{Proof of the Drift-Window Guarantee and Its Consequences}
\label{app:drift_proof}

This appendix proves \cref{thm:drift} and states and proves its two consequences, the excitation-dependent rate and the excitation-free functional bound. The windowed excitation condition and the corresponding result are stated here. Since $\Phi'(X, \theta) \in \mathbb{R}^{n \times p}$, the integrand of $\Psi$ has rank at most $n$, so a full-rank $\Psi$ requires the Jacobian row space to rotate through all of $\mathbb{R}^{p}$ within each window.  

\begin{definition}
\label{def:wpe}
The closed loop satisfies window persistence of excitation (WPE) if there exist $\varphi > 0$ and $t^{*} \geq t_{1}$ such that $\Psi(t) \succeq \varphi I_{p}$ for all $t \geq t^{*}$. Since $\Psi(t)$ is evaluated at $\hat{\theta}(t)$, this is a condition on the signals generated by the closed loop rather than on the state trajectory alone.
\end{definition}

Under this condition the rate of \cref{thm:drift} improves as follows.

\begin{theorem}
\label{thm:wpe}
Under the conditions of \cref{thm:drift} and \cref{def:wpe}, the bounds in \cref{eq:theta_bound} hold for $t \geq t^{*}$ with $t_{1}$ replaced by $t^{*}$ and the rate $\lambda_{d}$ replaced by $\lambda_{e} \triangleq 2\left(k_{1} + \left(\alpha - \tfrac{1}{2}\right)\varphi\right)\underline{\lambda}$.
\end{theorem}

Without any excitation, the same Lyapunov argument still bounds the time-averaged first-order functional error.

\begin{proposition}
\label{cor:functional}
Under the conditions of \cref{thm:drift},
$\limsup_{S \to \infty}\, \tfrac{1}{S}\int_{t_{1}}^{t_{1}+S} \tilde{\theta}(r)^{\top} \Psi(r)\, \tilde{\theta}(r)\, dr \leq c_{d}^{\infty}/(\alpha - \tfrac{1}{2})$.
\end{proposition}

The quantity $\tilde{\theta}(r)^{\top}\Psi(r)\tilde{\theta}(r)$ is the first-order functional error of the current model accumulated along the recent trajectory. Thus, up to the Taylor remainder and approximation error, \cref{cor:functional} bounds the time-averaged squared mismatch between $\Phi(\cdot,\hat{\theta})$ and $f$ without requiring excitation. For $p \gg n$, a bounded trajectory typically excites only a parameter subspace, making this functional statement more informative than a full-parameter claim.

\label{app:functional_proof}

Consider the Lyapunov function $V = \tfrac{1}{2}\tilde{\theta}^{\top} P(t)\tilde{\theta}$ with $P \triangleq \Gamma^{-1}$, which, by \cref{lem:gamma}, satisfies the sandwich inequality
\begin{equation}
\label{eq:sandwich}
\tfrac{1}{2\bar{\lambda}}\left\Vert \tilde{\theta} \right\Vert^{2} \leq V \leq \tfrac{1}{2\underline{\lambda}}\left\Vert \tilde{\theta} \right\Vert^{2}.
\end{equation}
Since $\dot{P} = -P\dot{\Gamma}P$, the update \cref{eq:Gamma_update} gives $\dot{P} = -\beta_{0}(P - \bar{\lambda}^{-1}I_{p}) + \varsigma_{\Gamma}\Psi$. Using $\dot{\tilde{\theta}} = -\dot{\hat{\theta}}$ and the projection property \cref{eq:proj_prop} with $y = \alpha\Xi_{f} - k_{\sigma}\hat{\theta}$,
\begin{equation}
\label{eq:Vdot_1}
\dot{V} = -\tilde{\theta}^{\top} P\dot{\hat{\theta}} + \tfrac{1}{2}\tilde{\theta}^{\top}\dot{P}\tilde{\theta}
\leq -\alpha\tilde{\theta}^{\top}\Xi_{f} + k_{\sigma}\tilde{\theta}^{\top}\hat{\theta} - \tfrac{\beta_{0}}{2}\tilde{\theta}^{\top}(P - \bar{\lambda}^{-1}I_{p})\tilde{\theta} + \tfrac{\varsigma_{\Gamma}}{2}\tilde{\theta}^{\top}\Psi\tilde{\theta}.
\end{equation}
Substituting \cref{eq:Xif_decomp}, using $\hat{\theta} = \theta^{*} - \tilde{\theta}$, $0 \leq \varsigma_{\Gamma} \leq 1$, $\Psi \succeq 0$, and $P - \bar{\lambda}^{-1}I_{p} \succeq 0$ (\cref{lem:gamma}),
\begin{equation}
\label{eq:Vdot_2}
\dot{V} \leq -\left(\alpha - \tfrac{1}{2}\right)\tilde{\theta}^{\top}\Psi\tilde{\theta} - k_{\sigma}\left\Vert \tilde{\theta} \right\Vert^{2} + k_{\sigma}\tilde{\theta}^{\top}\theta^{*} - \alpha\tilde{\theta}^{\top}\eta_{f}.
\end{equation}
Using Young's inequality in the form $ab \leq \tfrac{a^{2}}{4} + b^{2}$, the sigma-modification cross term is bounded as $k_{\sigma}\tilde{\theta}^{\top}\theta^{*} \leq \tfrac{k_{\sigma}}{4}\left\Vert \tilde{\theta} \right\Vert^{2} + k_{\sigma}\bar{\theta}^{2}$. Using \cref{eq:etaf_bound} and the projection bound $\left\Vert \tilde{\theta} \right\Vert \leq \bar{\tilde{\theta}}$,
\begin{equation}
\begin{split}
\alpha\left|\tilde{\theta}^{\top}\eta_{f}\right| &\leq \alpha c_{\eta 1}\bar{\tilde{\theta}}\left\Vert \tilde{\theta} \right\Vert^{2} + \alpha\left(c_{\eta 2}\bar{f}_{T}(t) + c_{\eta 3}\right)\left\Vert \tilde{\theta} \right\Vert \\
&\leq \left(\alpha c_{\eta 1}\bar{\tilde{\theta}} + \tfrac{k_{\sigma}}{4}\right)\left\Vert \tilde{\theta} \right\Vert^{2} + \tfrac{\alpha^{2}}{k_{\sigma}}\left(c_{\eta 2}\bar{f}_{T}(t) + c_{\eta 3}\right)^{2},
\end{split}
\end{equation}
where Young's inequality with parameter $k_{\sigma}/2$ was applied to the linear term. Collecting terms in \cref{eq:Vdot_2} yields
\begin{equation}
\label{eq:Vdot_final}
\dot{V} \leq -\left(\alpha - \tfrac{1}{2}\right)\tilde{\theta}^{\top}\Psi\tilde{\theta} - k_{1}\left\Vert \tilde{\theta} \right\Vert^{2} + c_{d}(t),
\end{equation}
with $k_{1} = \tfrac{k_{\sigma}}{2} - \alpha c_{\eta 1}\bar{\tilde{\theta}}$, which is positive by the gain condition of \cref{thm:drift} since $2\alpha c_{\eta 1}\bar{\tilde{\theta}} = \alpha T\bar{\upsilon}c_{\theta\theta}\bar{\tilde{\theta}}$, and $c_{d}(t)$ as in \cref{eq:cd_def}.

\textit{Proof of \cref{thm:drift}.} Dropping the first (non-positive) term in \cref{eq:Vdot_final} and using \cref{eq:sandwich} in the form $\left\Vert \tilde{\theta} \right\Vert^{2} \geq 2\underline{\lambda}V$,
\begin{equation}
\dot{V} \leq -\lambda_{d}V + c_{d}(t), \qquad \lambda_{d} = 2k_{1}\underline{\lambda}.
\end{equation}
The comparison lemma gives, for $t \geq t_{1}$,
\begin{equation}
V(t) \leq V(t_{1})e^{-\lambda_{d}(t - t_{1})} + \int_{t_{1}}^{t}e^{-\lambda_{d}(t - s)}c_{d}(s)\, ds.
\end{equation}
Since $\bar{f}_{T}$ is nonincreasing by the definition of $\bar{f}_{T}$, $c_{d}$ is nonincreasing with $\lim_{t\to\infty}c_{d}(t) = c_{d}^{\infty}$, so the integral is bounded by $c_{d}(t_{1})/\lambda_{d}$ and its limit superior by $c_{d}^{\infty}/\lambda_{d}$. Converting through \cref{eq:sandwich}, $\left\Vert \tilde{\theta}(t) \right\Vert^{2} \leq 2\bar{\lambda}V(t)$ and $V(t_{1}) \leq \left\Vert \tilde{\theta}(t_{1}) \right\Vert^{2}/(2\underline{\lambda})$, which yields both bounds in \cref{eq:theta_bound}. \hfill $\blacksquare$

\textit{Remark.} Since $c_{d}^{\infty} \geq k_{\sigma}\bar{\theta}^{2}$ and $k_{1} < k_{\sigma}/2$, the excitation-free radius in \cref{eq:theta_bound} satisfies $\tfrac{2\bar{\lambda}}{\lambda_{d}}c_{d}^{\infty} = \tfrac{\bar{\lambda}}{k_{1}\underline{\lambda}}c_{d}^{\infty} > 2\tfrac{\bar{\lambda}}{\underline{\lambda}}\bar{\theta}^{2}$, i.e., it exceeds $\bar{\theta}\sqrt{2\bar{\lambda}/\underline{\lambda}} > 2\bar{\theta}$ (for the diagnostic gains $\bar{\lambda}/\underline{\lambda} = 20$, more than $6\bar{\theta}$), whereas the projection alone gives $\left\Vert \tilde{\theta} \right\Vert \leq 2\bar{\theta} + \epsilon_{\theta}$. The same computation for \cref{eq:theta_bound_a} gives a radius of at least $\sqrt{2}\,\bar{\theta}$. This is the usual sigma-modification tradeoff: without excitation, \cref{thm:drift,thm:adjoint_uub} certify invariance and the exponential rate into the residual set, \cref{cor:functional} certifies the averaged functional error, and parameter localization requires excitation (\cref{thm:wpe}).

\textit{Proof of \cref{cor:functional}.} Retaining the first term in \cref{eq:Vdot_final} and dropping $-k_{1}\left\Vert \tilde{\theta} \right\Vert^{2}$, integration over $[t_{1}, t_{1}+S]$ gives
\begin{equation}
\begin{split}
\left(\alpha - \tfrac{1}{2}\right)&\int_{t_{1}}^{t_{1}+S}\tilde{\theta}(r)^{\top}\Psi(r)\tilde{\theta}(r)\, dr \\
&\leq V(t_{1}) - V(t_{1}+S) + \int_{t_{1}}^{t_{1}+S}c_{d}(r)\, dr \\
&\leq \tfrac{\bar{\tilde{\theta}}^{2}}{2\underline{\lambda}} + \int_{t_{1}}^{t_{1}+S}c_{d}(r)\, dr,
\end{split}
\end{equation}
where $V \geq 0$ and \cref{eq:sandwich} with $\left\Vert \tilde{\theta} \right\Vert \leq \bar{\tilde{\theta}}$ were used. Dividing by $S$ and letting $S \to \infty$, the first term vanishes and the Cesaro mean of the nonincreasing $c_{d}$ converges to $c_{d}^{\infty}$, which establishes the bound in \cref{cor:functional}. \hfill $\blacksquare$

\textit{Proof of \cref{thm:wpe}.} For $t \geq t^{*}$, $\tilde{\theta}^{\top}\Psi\tilde{\theta} \geq \varphi\left\Vert \tilde{\theta} \right\Vert^{2}$, so \cref{eq:Vdot_final} yields $\dot{V} \leq -\left(k_{1} + (\alpha - \tfrac{1}{2})\varphi\right)\left\Vert \tilde{\theta} \right\Vert^{2} + c_{d}(t) \leq -\lambda_{e}V + c_{d}(t)$, and the argument of \cref{thm:drift} applies verbatim with $\lambda_{d}$ replaced by $\lambda_{e}$. \hfill $\blacksquare$

\section{Adjoint Representation of the Trajectory Gradient}
\label{app:adjoint_proof}

This appendix states and proves the adjoint representation of the trajectory gradient used in \cref{sec:adjoint}.

\begin{theorem}
\label{thm:adjoint}
Let \cref{eq:containment_cond} hold, let $\hat{X}_{p}$ satisfy \cref{eq:virtual_system}, and let $\nu$ satisfy the backward equation in \cref{eq:adjoint_gradient} with $\nu(t) = 0_{2n}$. Then $\nabla_{\hat{\theta}}\mathcal{E}_{T}(t)$ equals the integral in \cref{eq:adjoint_gradient}.
\end{theorem}

Let $S(\tau) = \partial\hat{X}_{p}(\tau;t)/\partial\hat{\theta}$ denote the sensitivity, which satisfies $S' = \tfrac{\partial F}{\partial\chi}(\hat{X}_{p})S + \tfrac{\partial F}{\partial\theta}(\hat{X}_{p})$ with $S(t-T) = 0$, since the initial condition $X(t-T)$ does not depend on $\hat{\theta}$. By the chain rule,
\begin{equation}
\label{eq:gradE_S}
\nabla_{\hat{\theta}}\mathcal{E}_{T} = \int_{t-T}^{t} S(\tau)^{\top} e(\tau; t)\, d\tau.
\end{equation}
Consider $\int_{t-T}^{t}\nu^{\top}\left(S' - \tfrac{\partial F}{\partial\chi}S - \tfrac{\partial F}{\partial\theta}\right)d\tau = 0$. Integrating $\nu^{\top} S'$ by parts and using $\nu(t) = 0$ and $S(t-T) = 0$ eliminates the boundary terms, leaving
\begin{equation}
\int_{t-T}^{t}\left(-\dot{\nu}^{\top} - \nu^{\top}\tfrac{\partial F}{\partial\chi}\right)S\, d\tau = \int_{t-T}^{t}\nu^{\top}\tfrac{\partial F}{\partial\theta}\, d\tau.
\end{equation}
Selecting $\nu$ to satisfy the backward equation in \cref{eq:adjoint_gradient} makes the left integrand equal $e^{\top} S$, so \cref{eq:gradE_S} equals $\int_{t-T}^{t}\left(\tfrac{\partial F}{\partial\theta}\right)^{\top}\nu\, d\tau$, which is \cref{eq:adjoint_gradient}. \hfill $\blacksquare$

\section{Containment and Proof of the Trajectory Decomposition}
\label{app:containment_proof}
\label{app:traj_decomp_proof}

This appendix proves that the virtual trajectory stays in $\Omega$ and then proves the trajectory decomposition of \cref{lem:traj_decomp}. The containment lemma referenced in \cref{sec:adjoint} is stated first.

\begin{lemma}
\label{lem:containment}
Let \cref{assm:bounded_trajectory,assm:f_smooth} hold, let $\hat{\theta}(t) \in \Theta_{\epsilon}$, and let the window condition in \cref{eq:containment_cond} hold. Then $\hat{X}_{p}(\tau; t) \in \Omega$ and $\left\Vert e(\tau;t) \right\Vert < \varpi$ for all $\tau \in [t-T, t]$.
\end{lemma}

\textit{Proof of \cref{lem:containment}.} Let $\tau^{*} \triangleq \sup\{\tau \in [t-T, t] : \left\Vert e(s;t) \right\Vert \leq \varpi \; \forall s \in [t-T, \tau]\}$, which is well defined with $\tau^{*} > t-T$ by continuity and $e(t-T; t) = 0$. For $\tau \leq \tau^{*}$, both $\hat{X}_{p}(\tau)$ and the segment between $X(\tau)$ and $\hat{X}_{p}(\tau)$ lie in $\Omega$. On this interval,
\begin{equation}
\begin{split}
\dot{e} &= F(\hat{X}_{p}, \hat{\theta}, \tau) - F(X, \hat{\theta}, \tau) + F(X, \hat{\theta}, \tau) - \begin{bmatrix}\dot{x} \\ f(X) + g(X)u\end{bmatrix} \\
&= F(\hat{X}_{p}, \hat{\theta}, \tau) - F(X, \hat{\theta}, \tau) - B\left(f(X) - \Phi(X, \hat{\theta})\right),
\end{split}
\end{equation}
so, using the mean value bound $\left\Vert F(\hat{X}_{p},\hat{\theta},\tau) - F(X,\hat{\theta},\tau) \right\Vert \leq L_{F}\left\Vert e \right\Vert$ and, by \cref{eq:ufap} and the mean value theorem on the segment $[\hat{\theta}, \theta^{*}] \subset \Theta_{\epsilon}$, $\left\Vert f(X) - \Phi(X,\hat{\theta}) \right\Vert \leq \bar{\upsilon}\left\Vert \tilde{\theta} \right\Vert + \bar{\varepsilon} \leq \bar{\upsilon}\bar{\tilde{\theta}} + \bar{\varepsilon}$, the Gronwall inequality yields
\begin{equation}
\label{eq:e_apriori}
\left\Vert e(\tau; t) \right\Vert \leq \frac{e^{L_{F}(\tau - (t-T))} - 1}{L_{F}}\left(\bar{\upsilon}\bar{\tilde{\theta}} + \bar{\varepsilon}\right) \leq \kappa_{T}\left(\bar{\upsilon}\bar{\tilde{\theta}} + \bar{\varepsilon}\right) < \varpi,
\end{equation}
by \cref{eq:containment_cond}. If $\tau^{*} < t$, then continuity would require $\left\Vert e(\tau^{*}; t) \right\Vert = \varpi$, contradicting \cref{eq:e_apriori}; hence $\tau^{*} = t$ and the claim follows. \hfill $\blacksquare$

The operators used in \cref{lem:traj_decomp} are defined as follows. Let $X_{s}(\tau) \triangleq X(\tau) + s\, e(\tau; t)$ for $s \in [0,1]$ denote the segment between the measured and predicted states, and define the mean-value system matrix and its transition matrix as
\begin{equation}
\label{eq:Abar_def}
\bar{A}(\tau) \triangleq \int_{0}^{1} \frac{\partial F}{\partial \chi}\left(X_{s}(\tau),\, \hat{\theta}(t),\, \tau\right) ds, \qquad
\frac{\partial}{\partial \tau} \Phi_{A}(\tau, s) = \bar{A}(\tau)\, \Phi_{A}(\tau, s), \quad \Phi_{A}(s, s) = I_{2n}.
\end{equation}
Define the operators
\begin{align}
W(\tau) &\triangleq \int_{t-T}^{\tau} \Phi_{A}(\tau, s)\, B\, \Phi'(X(s), \hat{\theta}(t))\, ds, \label{eq:W_def}\\
d(\tau) &\triangleq \int_{t-T}^{\tau} \Phi_{A}(\tau, s)\, B \left[R(s) + \varepsilon(X(s))\right] ds, \label{eq:d_def}
\end{align}
where $R(s) \triangleq \Phi(X(s), \theta^{*}) - \Phi(X(s), \hat{\theta}(t)) - \Phi'(X(s), \hat{\theta}(t))\,\tilde{\theta}(t)$ denotes the Taylor remainder, which satisfies $\left\Vert R(s) \right\Vert \leq \tfrac{1}{2}c_{\theta\theta}\left\Vert \tilde{\theta}(t) \right\Vert^{2}$ by the smoothness of $\Phi$. Let $S(\tau) \triangleq \partial \hat{X}_{p}(\tau; t)/\partial \hat{\theta}$ denote the sensitivity of the virtual trajectory, which satisfies $S' = (\partial F/\partial \chi)(\hat{X}_{p})\, S + B\, \Phi'(\hat{X}_{p}, \hat{\theta})$ with $S(t-T) = 0$. Since $F(\hat{X}_{p}, \hat{\theta}) - F(X, \hat{\theta}) = \bar{A}\, e$ along the segment, the error dynamics are linear in $e$ with forcing $-B(\Phi'\tilde{\theta} + R + \varepsilon)$, which yields the following representation.

The constants in \cref{eq:etaa_bound} are $d_{2} \triangleq \tfrac{1}{2}c_{\theta\theta}\kappa_{T}$, $d_{0} \triangleq \kappa_{T}\bar{\varepsilon}$, $w_{e} \triangleq \bar{w} + d_{2}\bar{\tilde{\theta}}$, $D_{S} \triangleq T e^{L_{F}T}\left(c_{X\theta} + L_{FX}\, T\, \bar{\upsilon}\right)$,
\begin{equation}
b_{1} \triangleq T\left(\bar{w}\, d_{2} + D_{S}\, w_{e}^{2}\right), \qquad
b_{2} \triangleq 2\, T\, D_{S}\, w_{e}\, d_{0}, \qquad
b_{3} \triangleq T\, d_{0}\left(\bar{w} + D_{S}\, d_{0}\right).
\label{eq:b_defs}
\end{equation}

\textit{Proof of \cref{lem:traj_decomp}.} By \cref{lem:containment}, the segment $X(\tau) + s\, e(\tau;t)$, $s \in [0,1]$, remains in $\Omega$, so $\bar{A}$ in \cref{eq:Abar_def} is well defined with $\left\Vert \bar{A}(\tau) \right\Vert \leq L_{F}$, and by the fundamental theorem of calculus, $F(\hat{X}_{p},\hat{\theta},\tau) - F(X,\hat{\theta},\tau) = \bar{A}(\tau)e(\tau;t)$ exactly. Therefore
\begin{equation}
\dot{e} = \bar{A}(\tau)e - B\left(\Phi'(X(\tau),\hat{\theta})\tilde{\theta} + R(\tau) + \varepsilon(X(\tau))\right),
\end{equation}
using \cref{eq:taylor} and \cref{eq:ufap}, and variation of constants with $e(t-T;t) = 0$ gives the exact representation $e(\tau;t) = -\left(W(\tau)\tilde{\theta} + d(\tau)\right)$ with $W$ and $d$ as in \cref{eq:W_def}--\cref{eq:d_def}. Since $\left\Vert \Phi_{A}(\tau,s) \right\Vert \leq e^{L_{F}(\tau - s)}$,
\begin{equation}
\label{eq:Wd_bounds}
\left\Vert W(\tau) \right\Vert \leq \bar{\upsilon}\kappa_{T} = \bar{w}, \qquad \left\Vert d(\tau) \right\Vert \leq \kappa_{T}\left(\tfrac{1}{2}c_{\theta\theta}\left\Vert \tilde{\theta} \right\Vert^{2} + \bar{\varepsilon}\right) = d_{2}\left\Vert \tilde{\theta} \right\Vert^{2} + d_{0},
\end{equation}
and consequently, using $\left\Vert \tilde{\theta} \right\Vert \leq \bar{\tilde{\theta}}$,
\begin{equation}
\label{eq:esup_bound}
\left\Vert e(\tau;t) \right\Vert \leq \left(\bar{w} + d_{2}\bar{\tilde{\theta}}\right)\left\Vert \tilde{\theta} \right\Vert + d_{0} = w_{e}\left\Vert \tilde{\theta} \right\Vert + d_{0}.
\end{equation}

Next, the exact sensitivity $S$ and the operator $W$ are compared. Let $A^{v}(\tau) \triangleq \tfrac{\partial F}{\partial\chi}(\hat{X}_{p}(\tau),\hat{\theta},\tau)$ with transition matrix $\Phi_{v}(\tau,s)$, so that
\begin{equation*}
S(\tau) = \int_{t-T}^{\tau}\Phi_{v}(\tau,s)B\,\Phi'(\hat{X}_{p}(s),\hat{\theta})\, ds.
\end{equation*}
Since $A^{v}$ and $\bar{A}$ are evaluated at points within distance $\left\Vert e \right\Vert$ of one another inside $\Omega$, the bound on $\partial^{2}\Phi/\partial X\partial\theta$ and the definition of $L_{FX}$ give $\left\Vert A^{v}(\tau) - \bar{A}(\tau) \right\Vert \leq L_{FX}\left\Vert e(\tau;t) \right\Vert$. Writing $D(\tau,s) \triangleq \Phi_{v}(\tau,s) - \Phi_{A}(\tau,s)$, the difference satisfies $\partial_{\tau}D = A^{v}D + (A^{v} - \bar{A})\Phi_{A}$ with $D(s,s) = 0$, so
\begin{equation}
\left\Vert D(\tau,s) \right\Vert \leq \int_{s}^{\tau}e^{L_{F}(\tau - r)}\, L_{FX}\left\Vert e(r;t) \right\Vert\, e^{L_{F}(r-s)}\, dr \leq L_{FX}\, T\, e^{L_{F}T}\sup_{r}\left\Vert e(r;t) \right\Vert.
\end{equation}
Similarly, $\left\Vert \Phi'(\hat{X}_{p}(s),\hat{\theta}) - \Phi'(X(s),\hat{\theta}) \right\Vert \leq c_{X\theta}\left\Vert e(s;t) \right\Vert$. Combining,
\begin{equation}
\label{eq:SW_bound}
\left\Vert S(\tau) - W(\tau) \right\Vert \leq \int_{t-T}^{\tau}\left(\left\Vert D(\tau,s) \right\Vert\bar{\upsilon} + e^{L_{F}T}c_{X\theta}\left\Vert e(s;t) \right\Vert\right)ds \leq D_{S}\sup_{r \in [t-T,t]}\left\Vert e(r;t) \right\Vert,
\end{equation}
with $D_{S} = Te^{L_{F}T}\left(c_{X\theta} + L_{FX}T\bar{\upsilon}\right)$.

By \cref{thm:adjoint} and \cref{eq:gradE_S}, the adjoint signal satisfies exactly
\begin{equation}
\begin{split}
\Xi_{a} &= -\int_{t-T}^{t}S(\tau)^{\top} e(\tau;t)\, d\tau = \int_{t-T}^{t}S(\tau)^{\top}\left(W(\tau)\tilde{\theta} + d(\tau)\right)d\tau \\
&= H\tilde{\theta} + \int_{t-T}^{t}\left[W^{\top} d + (S - W)^{\top}\left(W\tilde{\theta} + d\right)\right]d\tau,
\end{split}
\end{equation}
which is \cref{eq:Xia_decomp} with $\eta_{a}$ equal to the last integral. Using \cref{eq:Wd_bounds}, \cref{eq:esup_bound}, \cref{eq:SW_bound}, $\left\Vert W\tilde{\theta} + d \right\Vert = \left\Vert e \right\Vert$, and $\left\Vert \tilde{\theta} \right\Vert \leq \bar{\tilde{\theta}}$ to fold cubic and quartic powers,
\begin{equation}
\begin{split}
\left\Vert \eta_{a} \right\Vert &\leq T\left[\bar{w}\left(d_{2}\left\Vert \tilde{\theta} \right\Vert^{2} + d_{0}\right) + D_{S}\left(w_{e}\left\Vert \tilde{\theta} \right\Vert + d_{0}\right)^{2}\right] \\
&\leq b_{1}\left\Vert \tilde{\theta} \right\Vert^{2} + b_{2}\left\Vert \tilde{\theta} \right\Vert + b_{3},
\end{split}
\end{equation}
with $b_{1}, b_{2}, b_{3}$ as in \cref{eq:b_defs}. Finally, $H$ is an integral of Gram matrices, hence positive semi-definite, with $\left\Vert H \right\Vert \leq T\bar{w}^{2}$. \hfill $\blacksquare$

\section{Proof of the Adjoint-Window Guarantee and the Window-Length Condition}
\label{app:adjoint_uub_proof}

This appendix proves \cref{thm:adjoint_uub} and the explicit window-length condition referenced in \cref{sec:adjoint}. The window-length corollary is stated first.

\begin{corollary}
\label{cor:window}
Let $L_{F} T \leq 1$. Then the gain condition of \cref{thm:adjoint_uub} holds whenever $T \leq \left(k_{\sigma}/(2\,\alpha\, C_{b})\right)^{1/3}$, where $C_{b}$ is the computable constant defined in \cref{app:adjoint_uub_proof}, obtained by majorizing $\kappa_{T} \leq e\, T$ and $e^{L_{F}T} \leq e$ in \cref{eq:b_defs}.
\end{corollary}

Consider $V_{a} = \tfrac{1}{2}\tilde{\theta}^{\top}\Gamma_{0}^{-1}\tilde{\theta}$, which satisfies $\tfrac{1}{2\bar{\lambda}}\left\Vert \tilde{\theta} \right\Vert^{2} \leq V_{a} \leq \tfrac{1}{2\underline{\lambda}}\left\Vert \tilde{\theta} \right\Vert^{2}$. Since $\Gamma_{0}$ is constant, using the projection property \cref{eq:proj_prop} with $y = \alpha\Xi_{a} - k_{\sigma}\hat{\theta}$ and the decomposition \cref{eq:Xia_decomp},
\begin{equation}
\begin{split}
\dot{V}_{a} &\leq -\alpha\tilde{\theta}^{\top} H\tilde{\theta} - \alpha\tilde{\theta}^{\top}\eta_{a} + k_{\sigma}\tilde{\theta}^{\top}\theta^{*} - k_{\sigma}\left\Vert \tilde{\theta} \right\Vert^{2} \\
&\leq -\alpha\tilde{\theta}^{\top} H\tilde{\theta} - k_{\sigma}\left\Vert \tilde{\theta} \right\Vert^{2} + k_{\sigma}\tilde{\theta}^{\top}\theta^{*} \\
&\quad + \alpha\left(b_{1}\bar{\tilde{\theta}} + b_{2}\right)\left\Vert \tilde{\theta} \right\Vert^{2} + \alpha b_{3}\left\Vert \tilde{\theta} \right\Vert,
\end{split}
\end{equation}
using \cref{eq:etaa_bound} and $\left\Vert \tilde{\theta} \right\Vert \leq \bar{\tilde{\theta}}$. Applying Young's inequality as in \cref{app:drift_proof}, $k_{\sigma}\tilde{\theta}^{\top}\theta^{*} \leq \tfrac{k_{\sigma}}{4}\left\Vert \tilde{\theta} \right\Vert^{2} + k_{\sigma}\bar{\theta}^{2}$ and $\alpha b_{3}\left\Vert \tilde{\theta} \right\Vert \leq \tfrac{k_{\sigma}}{4}\left\Vert \tilde{\theta} \right\Vert^{2} + \tfrac{\alpha^{2}b_{3}^{2}}{k_{\sigma}}$. Dropping the non-positive term $-\alpha\tilde{\theta}^{\top} H\tilde{\theta}$ yields
\begin{equation}
\dot{V}_{a} \leq -k_{1a}\left\Vert \tilde{\theta} \right\Vert^{2} + c_{a} \leq -\lambda_{a}V_{a} + c_{a},
\end{equation}
with $k_{1a} = \tfrac{k_{\sigma}}{2} - \alpha(b_{1}\bar{\tilde{\theta}} + b_{2}) > 0$ by the gain condition of \cref{thm:adjoint_uub} and $\lambda_{a} = 2k_{1a}\underline{\lambda}$. The comparison lemma and the sandwich inequality then give \cref{eq:theta_bound_a}. If WPE holds for $H$, that is $H(t) \succeq \varphi_{a}I_{p}$, the rate improves to $2(k_{1a} + \alpha\varphi_{a})\underline{\lambda}$ by retaining the dropped term. \hfill $\blacksquare$

\textit{Proof of \cref{cor:window}.} For $L_{F}T \leq 1$, the majorizations $\kappa_{T} = (e^{L_{F}T} - 1)/L_{F} \leq (e - 1)T < eT$, $e^{L_{F}T} \leq e$, and $T \leq 1/L_{F}$ hold, the first strictly for $T > 0$. Substituting into the definitions in \cref{eq:b_defs} gives $\bar{w} \leq e\bar{\upsilon}T$, $d_{2} \leq \tfrac{e c_{\theta\theta}}{2}T$, $d_{0} \leq e\bar{\varepsilon}T$, $w_{e} \leq e\, w_{1} T$ with $w_{1} \triangleq \bar{\upsilon} + \tfrac{c_{\theta\theta}\bar{\tilde{\theta}}}{2}$, and $D_{S} \leq eM_{S}T$ with $M_{S} \triangleq c_{X\theta} + L_{FX}\bar{\upsilon}/L_{F}$. Substituting into \cref{eq:b_defs} yields
\begin{equation}
b_{1}\bar{\tilde{\theta}} + b_{2} \leq \tfrac{e^{2}\bar{\upsilon}c_{\theta\theta}\bar{\tilde{\theta}}}{2}\, T^{3} + e^{3}M_{S}\, w_{1}\left(\bar{\tilde{\theta}}\, w_{1} + 2\bar{\varepsilon}\right) T^{4}
\leq C_{b}\, T^{3},
\end{equation}
where the factor $T^{4}$ was majorized by $T^{3}/L_{F}$ using $T \leq 1/L_{F}$, and
\begin{equation}
C_{b} \triangleq \frac{e^{2}\,\bar{\upsilon}\, c_{\theta\theta}\, \bar{\tilde{\theta}}}{2} + \frac{e^{3}\, M_{S}\, w_{1}}{L_{F}}\left(\bar{\tilde{\theta}}\, w_{1} + 2\bar{\varepsilon}\right).
\end{equation}
Since every term of $b_{1}\bar{\tilde{\theta}} + b_{2}$ contains at least one factor $\kappa_{T}$ (through $\bar{w}$, $d_{2}$, $d_{0}$, or $w_{e}$), and $\kappa_{T} < eT$ strictly, $b_{1}\bar{\tilde{\theta}} + b_{2} < C_{b}T^{3}$ for $T > 0$. Hence the window bound in \cref{cor:window} implies $2\alpha\left(b_{1}\bar{\tilde{\theta}} + b_{2}\right) < 2\alpha C_{b}T^{3} \leq k_{\sigma}$, which is the strict gain condition of \cref{thm:adjoint_uub}. \hfill $\blacksquare$

\section{Stored-Set Extension, Gauss--Newton Bound, and Proof of the NODE-CL Guarantee}
\label{app:stored_proof}
\label{app:stack_proof}

This appendix states and proves the stored-set extension of \cref{thm:drift}, the bound relating the Gauss--Newton matrix to the trajectory regressor, and \cref{thm:stack}.

\begin{proposition}
\label{cor:stored}
Let $\mathcal{D}(t) = \{(X_{j}, \hat{f}(\tau_{j}))\}_{j=1}^{N}$ be a set of $N$ stored pairs recorded at times $\tau_{j} \leq t$, with membership piecewise constant in $t$, and replace the integrals in \cref{eq:Xif_def}, \cref{eq:Psi_def}, and \cref{eq:Gamma_update} by $\Xi_{\mathcal{D}} \triangleq \tfrac{T}{N}\sum_{j=1}^{N}\Phi'(X_{j},\hat{\theta})^{\top}(\hat{f}(\tau_{j}) - \Phi(X_{j},\hat{\theta}))$ and $\Psi_{\mathcal{D}} \triangleq \tfrac{T}{N}\sum_{j=1}^{N}\Phi'(X_{j},\hat{\theta})^{\top}\Phi'(X_{j},\hat{\theta})$. Suppose every admitted sample satisfies $\left\Vert \tilde{f}(\tau_{j}) \right\Vert \leq \delta_{\mathcal{D}}$, which holds with $\delta_{\mathcal{D}} = \left\Vert z_{o}(t_{0}) \right\Vert e^{-k_{o}(t_{a} - t_{0})} + \bar{z}_{o}$ when admission starts at $t_{a} \geq t_{0}$. Let $t_{s} \geq t_{1}$ denote the time at which the stored-set law is activated, with the stack non-empty, $\hat{\theta}(t_{s}) \in \Theta_{\epsilon}$, and $\underline{\lambda} I_{p} \preceq \Gamma(t_{s}) \preceq \bar{\lambda} I_{p}$. Then \cref{lem:drift_decomp}, \cref{thm:drift}, and \cref{cor:functional} hold for $t \geq t_{s}$ with $t_{1}$ replaced by $t_{s}$ and with $\bar{f}_{T}(t)$ and $\bar{z}_{o}$ both replaced by $\delta_{\mathcal{D}}$, so that $c_{d}(t) = c_{d}^{\infty} = k_{\sigma}\bar{\theta}^{2} + \tfrac{\alpha^{2}}{k_{\sigma}}(c_{\eta 2}\delta_{\mathcal{D}} + c_{\eta 3})^{2}$ is constant.
\end{proposition}

A stored sample retains the observer error of its recording time, so the decaying window envelope becomes a uniform bound over admitted samples. The law of \citet{Hart2025} is a stored-set drift law of this type, whereas fixing $\Gamma$ recovers its constant-gain simplification, the DNN analogue of the stack of \citet{Chowdhary2011}. The additions analyzed here are the certified gain interval, gain-metric projection, and excitation-free functional bound; the finite-excitation convergence result of \citet{Chowdhary2011} is complementary.

\textit{Proof of \cref{cor:stored}.} The expansion in the proof of \cref{lem:drift_decomp} is performed sample by sample, so replacing the integral over $[t-T,t]$ by the average over $\mathcal{D}(t)$ scaled by $T$ gives $\Xi_{\mathcal{D}} = \Psi_{\mathcal{D}}\tilde{\theta} + \eta_{\mathcal{D}}$ with $\Psi_{\mathcal{D}} \succeq 0$, $\left\Vert \Psi_{\mathcal{D}} \right\Vert \leq T\bar{\upsilon}^{2}$, and $\left\Vert \eta_{\mathcal{D}} \right\Vert \leq c_{\eta 1}\left\Vert \tilde{\theta} \right\Vert^{2} + c_{\eta 2}\delta_{\mathcal{D}} + c_{\eta 3}$, since each stored $\hat{f}(\tau_{j})$ carries the observer error at its recording time and every admitted sample is bounded by $\delta_{\mathcal{D}}$; the bound on $\delta_{\mathcal{D}}$ under an admission time follows from \cref{eq:observer_bound}. The proofs of \cref{lem:gamma}, \cref{thm:drift}, and \cref{cor:functional} use only positive semi-definiteness and boundedness of the regressor and the form of the perturbation bound, with the constant $c_{d}^{\infty}$ in place of the nonincreasing $c_{d}(t)$, so they apply verbatim. Membership changes make $\Psi_{\mathcal{D}}$ piecewise continuous in $t$; $V$ is continuous and the inequality on $\dot{V}$ holds almost everywhere, so the comparison lemma applies unchanged. \hfill $\blacksquare$

The computable Gauss--Newton matrix is related to the trajectory regressor by the following bound.

\begin{lemma}
\label{lem:gn}
Under the conditions of \cref{lem:traj_decomp} at segment length $T_{s}$, $\left\Vert G_{\mathcal{S}} - H_{\mathcal{S}} \right\Vert \leq g_{2}\left\Vert \tilde{\theta} \right\Vert^{2} + g_{1}\left\Vert \tilde{\theta} \right\Vert + g_{0}$ with $g_{2} \triangleq T_{s}D_{S}^{2}w_{e}^{2}$, $g_{1} \triangleq 2T_{s}D_{S}w_{e}(\bar{w} + D_{S}d_{0})$, and $g_{0} \triangleq T_{s}D_{S}d_{0}(2\bar{w} + D_{S}d_{0})$, where $\bar{w}, w_{e}, d_{0}, D_{S}$ are the constants of \cref{app:traj_decomp_proof} at $T_{s}$.
\end{lemma}

\textit{Proof of \cref{lem:gn}.} By \cref{eq:SW_bound} and \cref{eq:esup_bound} applied to each segment at length $T_{s}$, $\left\Vert S_{j} - W_{j} \right\Vert \leq D_{S}(w_{e}\left\Vert \tilde{\theta} \right\Vert + d_{0}) \eqqcolon \Delta$, while $\left\Vert W_{j} \right\Vert \leq \bar{w}$ and hence $\left\Vert S_{j} \right\Vert \leq \bar{w} + \Delta$. Hence $\left\Vert S_{j}^{\top} S_{j} - W_{j}^{\top} W_{j} \right\Vert = \left\Vert (S_{j} - W_{j})^{\top} S_{j} + W_{j}^{\top}(S_{j} - W_{j}) \right\Vert \leq \Delta(2\bar{w} + \Delta)$, and integrating over a segment of length $T_{s}$ and averaging over segments gives $\left\Vert G_{\mathcal{S}} - H_{\mathcal{S}} \right\Vert \leq T_{s}\Delta(2\bar{w} + \Delta)$; expanding $\Delta$ in powers of $\left\Vert \tilde{\theta} \right\Vert$ yields the stated constants. \hfill $\blacksquare$

\textit{Proof of \cref{thm:stack}.} \cref{lem:gamma} holds because $G_{\mathcal{S}} \succeq 0$ and its proof uses nothing else. With $V = \tfrac{1}{2}\tilde{\theta}^{\top}\Gamma^{-1}\tilde{\theta}$ and \cref{eq:proj_prop},
\begin{equation*}
\dot{V} \leq -\alpha\tilde{\theta}^{\top} H_{\mathcal{S}}\tilde{\theta} - \alpha\tilde{\theta}^{\top}\eta_{\mathcal{S}} + k_{\sigma}\tilde{\theta}^{\top}\theta^{*} - k_{\sigma}\left\Vert \tilde{\theta} \right\Vert^{2} + \tfrac{1}{2}\varsigma_{\Gamma}\tilde{\theta}^{\top} H_{\mathcal{S}}\tilde{\theta} + \tfrac{1}{2}\varsigma_{\Gamma}\tilde{\theta}^{\top}(G_{\mathcal{S}} - H_{\mathcal{S}})\tilde{\theta},
\end{equation*}
where the last two terms come from $\tfrac{1}{2}\tilde{\theta}^{\top}\tfrac{d}{dt}(\Gamma^{-1})\tilde{\theta}$ with the term $-\tfrac{\beta_{0}}{2}\tilde{\theta}^{\top}(\Gamma^{-1} - \bar{\lambda}^{-1}I_{p})\tilde{\theta} \leq 0$ dropped. Since $0 \leq \varsigma_{\Gamma} \leq 1$, the $H_{\mathcal{S}}$ terms combine into $-(\alpha - \tfrac{1}{2})\tilde{\theta}^{\top} H_{\mathcal{S}}\tilde{\theta}$, and \cref{lem:gn} with $\left\Vert \tilde{\theta} \right\Vert \leq \bar{\tilde{\theta}}$ bounds the mismatch term by $\tfrac{1}{2}(g_{2}\bar{\tilde{\theta}}^{2} + g_{1}\bar{\tilde{\theta}} + g_{0})\left\Vert \tilde{\theta} \right\Vert^{2}$. Bounding the remaining terms as in the proof of \cref{thm:adjoint_uub} and using $H_{\mathcal{S}}\succeq\varphi_s I_p$ gives
\begin{equation*}
\dot V\leq-k_s(\varphi_s)\left\Vert \tilde{\theta} \right\Vert^{2}+c_a
\leq-2k_s(\varphi_s)\underline{\lambda}V+c_a
=-\lambda_s(\varphi_s)V+c_a.
\end{equation*}
The comparison lemma and $\left\Vert \tilde{\theta} \right\Vert^{2}/(2\bar{\lambda})\leq V\leq\left\Vert \tilde{\theta} \right\Vert^{2}/(2\underline{\lambda})$ yield \cref{eq:stack_bound}. \hfill $\blacksquare$

\section{Simulation Details, Additional Ablations, and Verification of the Certified Inequalities}
\label{app:sim_details}

The diagnostic results are reported in \cref{tab:s12} and \cref{fig:curves,fig:noise,fig:s1_subspace,fig:s1_internals,fig:s2_rollout,fig:ablation,fig:ablation_app}; the full benchmark results in \cref{tab:bench_full} and \cref{fig:bench_curves,fig:bench_rollout}.

\textbf{Plants and excitation.} In S1, the plant is $\ddot{x} = \Phi(X, \theta_{\mathrm{tr}}) + u$ with $n = m = 1$, where $\Phi$ is a single-hidden-layer $\tanh$ network with input $X \in \mathbb{R}^{2}$, six hidden neurons, and scalar output ($p = 25$), and $\theta_{\mathrm{tr}}$ is drawn per seed with hidden weights and biases from $\mathcal{N}(0, 1)$ and $\mathcal{N}(0, 0.25)$ and output weights and bias from $\mathcal{N}(0, 0.64)$ and $\mathcal{N}(0, 0.04)$. The input is $u = \sum_{k=1}^{7} a_{k}\sin(2\pi \omega_{k} t + \phi_{k}) - x - 1.5\dot{x}$ with $\omega_{k} \in \{0.20, 0.31, 0.47, 0.71, 1.03, 1.51, 2.23\}\,$Hz, amplitudes $a_{k} = 1.1\,\xi_{k}$ with $\xi_{k} \sim \mathcal{U}[0.5, 1]$, and phases $\phi_{k} \sim \mathcal{U}[0, 2\pi]$. The estimate is initialized at $\hat{\theta}(t_{0}) = \theta_{\mathrm{tr}} + \delta$ with $\delta$ Gaussian and rescaled to $\left\Vert \delta \right\Vert = 0.4\left\Vert \theta_{\mathrm{tr}} \right\Vert$, and $\bar{\theta} = 2\left\Vert \theta_{\mathrm{tr}} \right\Vert$, $\epsilon_{\theta} = 0.5$. In S2, the manipulator model is $M(q)\ddot{q} + C(q, \dot{q})\dot{q} + F_{d}\dot{q} = \tau$ with $M = [p_{1} + 2p_{3}c_{2},\; p_{2} + p_{3}c_{2};\; p_{2} + p_{3}c_{2},\; p_{2}]$, $C = [-p_{3}s_{2}\dot{q}_{2},\; -p_{3}s_{2}(\dot{q}_{1} + \dot{q}_{2});\; p_{3}s_{2}\dot{q}_{1},\; 0]$, $c_{2} = \cos q_{2}$, $s_{2} = \sin q_{2}$, $p_{1} = 3.473$, $p_{2} = 0.196$, $p_{3} = 0.242$, and $F_{d} = \mathrm{diag}(5.3, 1.1)$. The input is a five-component multisine per joint with frequencies $\{0.13, 0.29, 0.53, 0.97, 1.61\}\,$Hz and $\{0.17, 0.37, 0.67, 1.19, 1.87\}\,$Hz, amplitudes $4\xi_{k}$ and $2\xi_{k}$, and the baseline $-3q - 3\dot{q}$. The network has input $X \in \mathbb{R}^{4}$, twelve hidden neurons, and two outputs ($p = 86$), initialized with entries from $\mathcal{N}(0, 0.0025)$; $\bar{\theta} = 40$ and $\epsilon_{\theta} = 2$. The held-out excitation uses a different draw of amplitudes and phases.

\textbf{Gains and discretization (S1 and S2).} The observer gains are $\alpha_{o} = k_{f} = 200$, the learning gains are $\alpha = 2$ (drift, single-step, CL) and $\alpha = 6$ (adjoint), and the gain matrix parameters of the windowed drift law are $\Gamma_{0} = 5 I_{p}$, $\bar{\lambda} = 10$, $\underline{\lambda} = 0.5$, $\beta_{0} = 0.05$, and $k_{\sigma} = 10^{-4}$ (the CL-LS and stored-segment gains are given under \emph{Stored-set methods} below). The plant is integrated by fourth-order Runge-Kutta at $\Delta t = 1\,$ms. The observer is integrated by fourth-order Runge-Kutta over each sampling interval from the sampled velocity, with the velocity extrapolated within the interval as $\dot{x}(t_{k}) + (\hat{f}(t_{k}) + g(X(t_{k}))u(t_{k}))(\tau - t_{k})$ (predictive hold); the zero-order hold variant in \cref{fig:ablation}(b) uses $\dot{x}(t_{k})$ throughout the interval. The drift-law, single-step, and CL updates are applied by forward Euler at $1\,$ms in S1 and $5\,$ms in S2, and the adjoint and Adam updates at $10\,$ms in both studies, with the gradient held constant between updates. The forward and backward solves of the adjoint use second-order integration on the $10\,$ms buffer grid.

\textbf{Baselines (S1 and S2).} The single-step baseline uses $\Xi = T\,\Phi'(X(t), \hat{\theta})^{\top} E(t;t)$ so that its gradient magnitude matches that of the windowed laws. CL stores up to $100$ pairs $(X_{j}, \hat{f}(\tau_{j}))$; a new pair is admitted when the relative change of the Jacobian $\Phi'(X, \hat{\theta})$ from the last stored point exceeds $0.1$, and once the stack is full it replaces the stored point whose replacement most increases the minimum singular value of the stacked regressor, evaluated every $50\,$ms in S1 and $200\,$ms in S2 with the current $\hat{\theta}$; its gradient is $\Xi = (T/N)\sum_{j=1}^{N}\Phi'(X_{j},\hat{\theta})^{\top}(\hat{f}(\tau_{j}) - \Phi(X_{j},\hat{\theta}))$ with a constant gain $\Gamma_{0}$, and the stack-selection check, which requires a singular value decomposition per stored point, is not included in its reported update cost. Online Adam uses $\beta_{1} = 0.9$, $\beta_{2} = 0.999$, and learning rates $10^{-2}$ (S1) and $3 \times 10^{-3}$ (S2), selected from $\{3 \times 10^{-4}, 10^{-3}, 3 \times 10^{-3}, 10^{-2}\}$ on seed $0$; the single-step gain $\alpha = 2$ was selected on the pendulum tuning seed.

\textbf{Evaluation (S1 and S2).} The evaluation states are $400$ samples, equally spaced in time, from the trajectory after $t = 20\,$s ($120\,$s runs) or $15\,$s ($60\,$s runs). The excited subspace in \cref{fig:s1_subspace}(b) is a diagnostic computed from the teacher Jacobian: the span of eigenvectors of the average of $\Phi'(X, \theta_{\mathrm{tr}})^{\top}\Phi'(X, \theta_{\mathrm{tr}})$ over the evaluation states with eigenvalues above $10^{-3}$ of the largest; its invariance under the adaptation is not established, and the distinction drawn is between excited and weakly excited directions rather than a strict rank deficiency.

\begin{figure}[tbp]
\centering
\includegraphics[width=\textwidth]{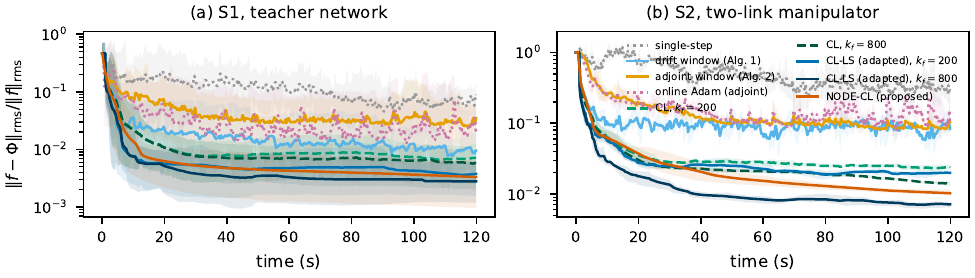}
\caption{Normalized functional error over time (three seeds, geometric mean with seed range) for (a) S1 and (b) S2.}
\label{fig:curves}
\end{figure}

\begin{table}[tbp]
\caption{Diagnostic studies S1 and S2 (three seeds, $120\,$s; these runs use the scalar bounded-gain forgetting of \citet{SlotineLi1989} in the least-squares laws): terminal normalized functional error (geometric mean, seed range); final parameter error ratio (S1); held-out $4\,$s prediction error (S2; initially $14.3$); update cost in ms on one CPU core, excluding stack selection. All rows except online Adam are sampled implementations of continuous-time laws covered by the analysis under its stated assumptions and gain conditions.}
\label{tab:s12}
\centering
\scriptsize
\setlength{\tabcolsep}{2.75pt}
\begin{tabular}{lcccccc}
\toprule
 & \multicolumn{2}{c}{S1 (teacher, $p=25$)} & \multicolumn{2}{c}{S2 (manipulator, $p=86$)} & \multicolumn{2}{c}{cost (ms)} \\
\cmidrule(lr){2-3}\cmidrule(lr){4-5}\cmidrule(lr){6-7}
Method & func. error [range] & $\left\Vert \tilde{\theta}(t_f) \right\Vert/\left\Vert \tilde{\theta}(0) \right\Vert$ & func. error [range] & pred. ($4\,$s) & S1 & S2 \\
\midrule
single-step & $6.8$e$-2$ [$3.1$e$-2$,$1.1$e$-1$] & $0.89$ & $3.2$e$-1$ [$2.2$e$-1$,$5.1$e$-1$] & $0.293$ & $0.03$ & $0.04$ \\
drift window (Alg.\,1) & $9.3$e$-3$ [$3.7$e$-3$,$1.7$e$-2$] & $0.80$ & $7.8$e$-2$ [$6.1$e$-2$,$1.1$e$-1$] & $0.119$ & $0.15$ & $0.81$ \\
adjoint window (Alg.\,2) & $2.9$e$-2$ [$6.8$e$-3$,$7.3$e$-2$] & $0.85$ & $8.9$e$-2$ [$8.0$e$-2$,$9.9$e$-2$] & $0.122$ & $1.4$ & $2.5$ \\
online Adam (adjoint) & $2.4$e$-2$ [$9.0$e$-3$,$7.7$e$-2$] & $1.31$ & $1.2$e$-1$ [$9.4$e$-2$,$1.8$e$-1$] & $0.157$ & $1.4$ & $2.5$ \\
CL, $k_{f}=200$ & $6.9$e$-3$ [$2.6$e$-3$,$1.2$e$-2$] & $0.83$ & $2.3$e$-2$ [$2.1$e$-2$,$2.6$e$-2$] & $0.0375$ & $0.06$ & $0.14$ \\
CL, $k_{f}=800$ & $5.9$e$-3$ [$2.1$e$-3$,$9.9$e$-3$] & $0.83$ & $1.4$e$-2$ [$1.4$e$-2$,$1.5$e$-2$] & $0.0399$ & $0.06$ & $0.14$ \\
CL-LS (adapted), $k_{f}=200$ & $3.7$e$-3$ [$1.6$e$-3$,$6.0$e$-3$] & $0.93$ & $2.0$e$-2$ [$1.8$e$-2$,$2.1$e$-2$] & $0.0232$ & $0.14$ & $0.78$ \\
CL-LS (adapted), $k_{f}=800$ & $\mathbf{2.8}$\textbf{e}$\mathbf{-3}$ [$1.2$e$-3$,$6.2$e$-3$] & $0.92$ & $\mathbf{7.1}$\textbf{e}$\mathbf{-3}$ [$6.3$e$-3$,$7.7$e$-3$] & $0.0249$ & $0.14$ & $0.78$ \\
NODE-CL (proposed) & $3.4$e$-3$ [$9.8$e$-4$,$6.6$e$-3$] & $0.80$ & $1.0$e$-2$ [$9.9$e$-3$,$1.2$e$-2$] & $0.0374$ & $2.3$ & $15$ \\
NODE-CL, equal memory & $4.8$e$-3$ [$1.7$e$-3$,$8.5$e$-3$] & $\mathbf{0.80}$ & $1.7$e$-2$ [$1.4$e$-2$,$2.1$e$-2$] & $0.0279$ & $0.98$ & $5.1$ \\
\bottomrule
\end{tabular}
\end{table}

\begin{figure}[tbp]
\centering
\includegraphics[width=\textwidth]{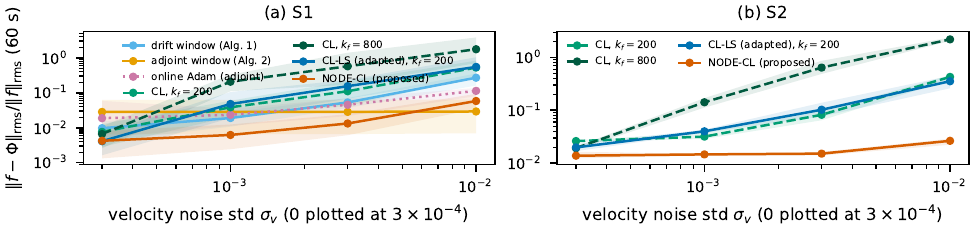}
\caption{Terminal normalized functional error after $60\,$s against the standard deviation of white velocity measurement noise (three seeds, geometric mean with seed range) for (a) S1 and (b) S2.}
\label{fig:noise}
\end{figure}

\begin{table}[tbp]
\caption{Benchmark suite, full results (five evaluation seeds, medians, $60\,$s): terminal normalized functional error on visited states, held-out $1\,$s error with seed range, held-out $4\,$s error, for clean (c) and noisy (n, $\sigma_{v} = 3\times10^{-3}$) measurements; mean cost per update and amortized stack-selection cost per simulated second (ms, one CPU core, NumPy). The $10^{-3}$ and $10^{-2}$ levels are shown in \cref{fig:bench_noise}; the $10^{-2}$ level was not run for the acrobot and reacher.}
\label{tab:bench_full}
\centering
\scriptsize
\setlength{\tabcolsep}{3.5pt}
\begin{tabular}{llcccccccc}
\toprule
Domain & Method & func (c) & $1\,$s (c) & $4\,$s (c) & func (n) & $1\,$s (n) & $4\,$s (n) & upd ms & sel ms/s \\
\midrule
pendulum & single-step & $1.0e+00$ & $1.643$ [1.00, 2.05] & $1.27$ & $8.7e-01$ & $1.650$ & $1.24$ & $0.03$ & $0$ \\
pendulum & CL & $2.5e-02$ & $0.315$ [0.17, 0.73] & $0.29$ & $1.3e-01$ & $0.609$ & $0.70$ & $0.09$ & $44$ \\
pendulum & CL-LS & $2.3e-02$ & $0.195$ [0.03, 0.25] & $0.11$ & $2.0e-01$ & $1.642$ & $0.90$ & $0.49$ & $40$ \\
pendulum & NODE-replay & $5.0e-03$ & $0.183$ [0.06, 0.26] & $0.12$ & $1.5e-02$ & $0.141$ & $0.11$ & $4.77$ & $188$ \\
pendulum & NODE-CL & $2.7e-03$ & $0.035$ [0.01, 0.04] & $0.02$ & $1.2e-02$ & $0.038$ & $0.03$ & $5.32$ & $186$ \\
cartpole & single-step & $4.6e-01$ & $3.287$ [1.65, 4.44] & $7.50$ & $4.5e-01$ & $3.145$ & $8.28$ & $0.04$ & $0$ \\
cartpole & CL & $7.0e-02$ & $0.623$ [0.42, 0.95] & $2.71$ & $1.2e-01$ & $1.363$ & $7.96$ & $0.16$ & $196$ \\
cartpole & CL-LS & $2.2e-02$ & $0.251$ [0.21, 0.95] & $1.11$ & $1.1e-01$ & $1.984$ & $28.61$ & $1.91$ & $186$ \\
cartpole & NODE-replay & $1.6e-02$ & $0.311$ [0.25, 0.49] & $1.81$ & $2.8e-02$ & $0.741$ & $2.57$ & $12.29$ & $495$ \\
cartpole & NODE-CL & $2.1e-02$ & $0.381$ [0.29, 0.57] & $1.95$ & $1.7e-02$ & $0.302$ & $2.09$ & $13.54$ & $489$ \\
acrobot & single-step & $5.5e+00$ & $1.621$ [0.56, 9.48] & $4.79$ & $3.9e+00$ & $1.232$ & $2.72$ & $0.04$ & $0$ \\
acrobot & CL & $6.4e-01$ & $3.615$ [0.92, 5.63] & $16.24$ & $5.4e-01$ & $4.385$ & $19.50$ & $0.16$ & $224$ \\
acrobot & CL-LS & $5.9e-02$ & $1.409$ [1.21, 4.33] & $7.04$ & $2.4e-01$ & $2.015$ & $7.64$ & $2.34$ & $225$ \\
acrobot & NODE-replay & $3.3e-01$ & $1.002$ [0.75, 2.55] & $3.79$ & $3.3e-01$ & $2.129$ & $6.15$ & $14.45$ & $598$ \\
acrobot & NODE-CL & $1.5e-01$ & $1.545$ [0.31, 1.79] & $5.09$ & $2.0e-01$ & $1.279$ & $3.30$ & $15.95$ & $596$ \\
reacher & single-step & $1.3e+00$ & $1.509$ [0.89, 4.02] & $1.78$ & $1.3e+00$ & $1.516$ & $1.78$ & $0.04$ & $0$ \\
reacher & CL & $6.5e-01$ & $1.007$ [0.83, 1.63] & $0.90$ & $6.4e-01$ & $0.976$ & $0.98$ & $0.16$ & $236$ \\
reacher & CL-LS & $1.6e-01$ & $0.064$ [0.05, 0.12] & $0.24$ & $3.6e-01$ & $0.215$ & $0.60$ & $2.36$ & $230$ \\
reacher & NODE-replay & $2.4e-01$ & $0.077$ [0.03, 0.12] & $0.11$ & $2.3e-01$ & $0.088$ & $0.18$ & $14.43$ & $604$ \\
reacher & NODE-CL & $2.4e-01$ & $0.069$ [0.04, 0.08] & $0.12$ & $2.4e-01$ & $0.075$ & $0.12$ & $15.94$ & $589$ \\
\bottomrule
\end{tabular}
\end{table}

\begin{figure}[tbp]
\centering
\includegraphics[width=\textwidth]{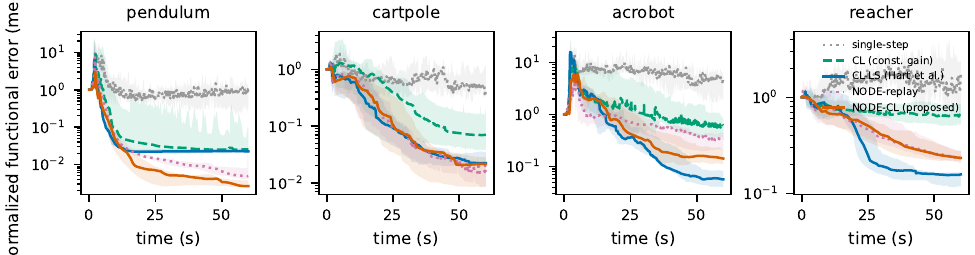}
\caption{Benchmark suite learning curves: normalized functional error on the visited states against streaming time (five evaluation seeds, median with seed range, clean measurements).}
\label{fig:bench_curves}
\end{figure}

\begin{figure}[tbp]
\centering
\includegraphics[width=\textwidth]{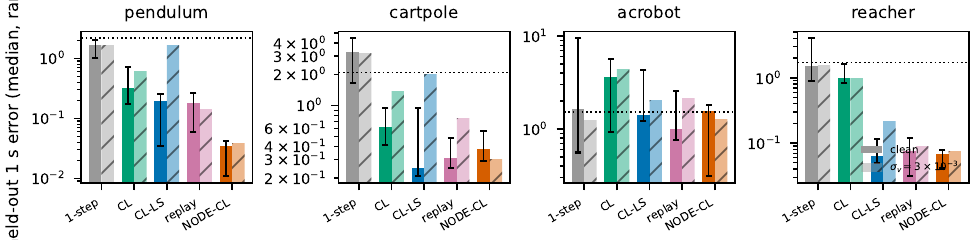}
\caption{Held-out $1\,$s prediction error per domain (five evaluation seeds, median with seed range), clean (solid) and under velocity noise $3\times10^{-3}$ (hatched); the dotted line is the initial model. Learning curves are given in \cref{fig:bench_curves}.}
\label{fig:bench_rollout}
\end{figure}

\textbf{Benchmark suite details.} Domains and tasks are pendulum (swingup), cartpole (swingup), acrobot (swingup), and reacher (easy) from \citet{Tassa2018}, integrated by MuJoCo at $1\,$ms with the suite's models unchanged; all are contact free. State dimensions $2n$ are $2$, $4$, $4$, and $4$, with $1$, $1$, $1$, and $2$ actuators; the network input has dimension $4$, $6$, $7$, and $7$ after the $(\cos,\sin)$ features (only unlimited hinges receive the feature), the hidden layer has $12$, $16$, $16$, and $16$ $\tanh$ neurons, and $p = 73$, $146$, $162$, and $162$. The excitation is a five-component multisine per actuator with frequencies uniform in $[0.1, 1.5]\,$Hz, random phases, and amplitudes of $0.45$, $0.5$, $0.25$, and $0.6$ of the actuator range, plus a velocity (and, for the cartpole slider, position) feedback baseline that keeps the trajectory bounded; the held-out record uses a separate excitation seed and $12\,$s, from which eight initial conditions spaced $1.2\,$s apart are used. The network is initialized from $\mathcal{N}(0, 0.0025)$ with $\bar{\theta} = 60$, $\epsilon_{\theta} = 3$.

\textbf{Benchmark suite gains.} Observer gains are $\alpha_{o} = k_{f} = 200$ with the predictive hold. Single-step uses $\alpha = 0.1$ at a $1\,$ms update (selected on the pendulum from $\{0.5, 0.1, 0.02\}$; larger gains oscillate); CL uses $\alpha = 2$, $\Gamma = 5I_{p}$, a stack of $100$ pairs admitted at $0.2\,$s checks with relative-Jacobian tolerance $0.1$, at a $5\,$ms update; CL-LS adds \cref{eq:Gamma_update} in information form with $\beta_{0} = 0.3$, $\bar{\lambda} = 10^{3}$, $\underline{\lambda} = 10^{-3}$, $k_{\sigma} = 10^{-6}$; NODE-CL uses segments of $0.05\,$s ($200$ on the pendulum, $100$ elsewhere), admission checks every $0.1\,$s, $\alpha = 2$, $\beta_{0} = 0.3$, $\bar{\lambda} = 10^{6}$, $\underline{\lambda} = 10^{-5}$, $k_{\sigma} = 10^{-9}$, $\Gamma_{0} = 10^{4}I_{p}$, at a $20\,$ms update; NODE-replay uses the same stack and update interval with learning rate $3 \times 10^{-2}$ (selected on S2 from $\{10^{-3}, 3\times10^{-3}, 10^{-2}, 3\times10^{-2}\}$). The forgetting rate $\beta_{0} = 0.3$ for both least-squares laws was selected on cartpole seed $0$ from $\{0.1, 0.3, 1\}$; with the scalar forgetting of \citet{SlotineLi1989} the NODE-CL error on that seed was $0.109$ against $0.037$ with the direction-wise forgetting of \cref{eq:Gamma_update}.

\textbf{Benchmark suite memory and the swimmer.} CL and CL-LS store $100$ input-output pairs ($100(2n + m + n)$ scalars) plus a $p \times p$ gain for CL-LS; NODE-CL stores $N_{s}$ segments of six samples ($6N_{s}(2n + m)$ scalars), a cache of $N_{s}$ terminal sensitivities ($2nN_{s}p$ scalars), and the $p \times p$ gain. The swimmer (six links, $2n = 16$, $m = 5$, $p = 356$ with $12$ hidden neurons, translation coordinates excluded from the input) was run with $30$ segments, $0.1\,$s updates, and $0.5\,$s admission checks; after $60\,$s every method, including NODE-CL, has normalized functional error above $1$ and held-out error above the initial model, so the domain is not discriminating at this data budget and is reported only here.

\textbf{Relation to the concurrent-learning law of \citet{Hart2025}.} CL-LS is not claimed as a new algorithm; it is an adapted baseline. The prior law updates all DNN layers by $\dot{\hat{\theta}} = \mathrm{proj}(\Gamma(\Phi'(X,\hat{\theta})^{\top} r - \gamma_{1}\sum_{i=1}^{N}\Phi'(X_{i},\hat{\theta})^{\top}(u_{i} - \hat{u}_{i}) - \gamma_{2}\hat{\theta}))$ with $\tfrac{d}{dt}\Gamma^{-1} = -\beta\Gamma^{-1} + \gamma_{1}\sum_{i=1}^{N}\Phi'(X_{i},\hat{\theta})^{\top}\Phi'(X_{i},\hat{\theta})$ frozen outside an eigenvalue band, where $u_{i} - \hat{u}_{i}$ is the stored residual between the applied input and its reconstruction from a state-derivative observer, and the stack is admitted after the observer settling time; its Lyapunov derivative bounds the gain-dynamics term by half the gradient term, as in the proof of \cref{thm:drift}. CL-LS in this paper restricts that mechanism to identification (no tracking term), with the windowed or stored residual $\hat{f}(\tau_{j}) - \Phi(X_{j},\hat{\theta})$ normalized by $T/N$, the direction-wise forgetting and Lipschitz gate of \cref{eq:Gamma_update} in place of the frozen band, which yields the certified bounds of \cref{lem:gamma}, and the projection in the $\Gamma$ metric with the inequality of \cref{lem:proj}; the constant-gain CL row is the same DNN stack with $\Gamma$ fixed, i.e., the point-stack rule of \citet{Chowdhary2011} applied to the DNN Jacobian, which is a simplification of the law of \citet{Hart2025} rather than a separate prior method. Neither law is new in its residual; the contributions of this paper on the drift side are the certified gain and projection properties and the excitation-free guarantee, and on the trajectory side the decomposition of \cref{lem:traj_decomp} and the law of \cref{thm:stack}.

\textbf{Stored-set methods (S1 and S2).} The following describes the diagnostic studies; the benchmark-suite settings are given above. CL-LS uses the CL stack (the same admission rule, tolerance, and check interval) with the scalar bounded-gain forgetting variant of the least-squares update, $\tfrac{d}{dt}\Gamma^{-1} = \varsigma_{\Gamma}\Psi_{\mathcal{D}} - \beta\Gamma^{-1}$ with $\beta = \beta_{0}(1 - \lambda_{\max}(\Gamma)/\bar{\lambda})$ \citep{SlotineLi1989}, implemented by forward Euler at the update interval, with $\beta_{0} = 1$, $\bar{\lambda} = 10^{3}$, $\underline{\lambda} = 10^{-3}$, $k_{\sigma} = 10^{-6}$, and $\Gamma_{0} = 5I_{p}$. The stored-segment law uses $N_{s} = 200$ segments of $T_{s} = 0.05\,$s (six samples on the $10\,$ms buffer grid), or $100$ segments of $0.03\,$s (four samples) in the equal-memory variant, rolled out in a batch with second-order integration of $\hat{X}_{p,j}$ and $S_{j}$ and trapezoidal quadrature of \cref{eq:stack_defs}; candidates are the most recent $T_{s}$ of the buffer, tested every $0.1\,$s ($0.2\,$s for the equal-memory variant) by the singular-value maximizing rule on the stacked terminal sensitivities, evaluated on the Gram matrix with rank-$2n$ updates; the gain uses $\alpha = 2$, $\beta_{0} = 1$, $\bar{\lambda} = 10^{6}$, $\underline{\lambda} = 10^{-5}$, $k_{\sigma} = 10^{-9}$, $\Gamma_{0} = 10^{4}I_{p}$, and the update interval is $20\,$ms. The large $\bar{\lambda}$ and small $k_{\sigma}$ reflect the $T_{s}^{3}$ scale of $G_{\mathcal{S}}$, on which $\Gamma \approx \beta G_{\mathcal{S}}^{-1}$ is correspondingly large; the leakage rate $k_{\sigma}\lambda_{\max}(\Gamma)$ is at most $10^{-3}\,\mathrm{s}^{-1}$ in all runs. The forward-sensitivity gradient agrees with the adjoint gradient of \cref{thm:adjoint} to a relative error of $1.7 \times 10^{-3}$ on a $0.25\,$s segment, which is the discretization level. The observer gain $k_{f} = 800$ runs use $\alpha_{o} = k_{f} = 800$ with the predictive hold. Velocity measurement noise is white Gaussian noise added to each $1\,$ms velocity sample used by the algorithms; the plant itself is noise free. The excitation-richness ablation keeps the first $K$ multisine components at equal total power, and the width ablation uses $6$, $12$, and $24$ hidden neurons in S2.

\begin{figure}[tbp]
\centering
\includegraphics[width=\textwidth]{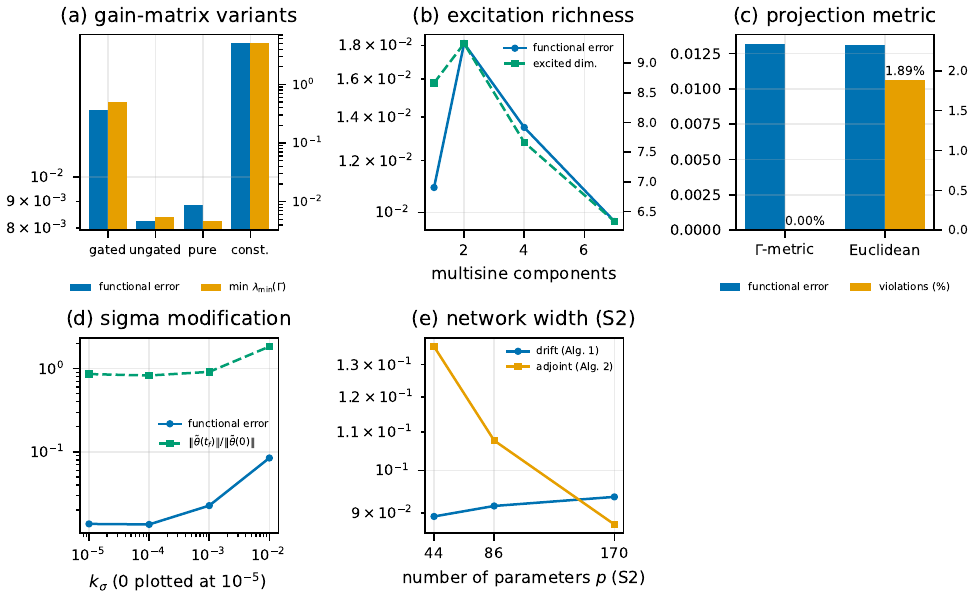}
\caption{Additional ablations (three seeds, $60\,$s; geometric means; violation rates and excited dimensions are arithmetic means, and zero violations are drawn as zero). (a) Gain-matrix variants at $\alpha = 2$: terminal functional error and the smallest eigenvalue of $\Gamma$ attained during the run. (b) Excitation richness: terminal functional error and the dimension of the excited subspace against the number of multisine components at equal total power. (c) Projection metric on a tight parameter set ($\bar{\theta} = \left\Vert \theta_{\mathrm{tr}} \right\Vert$, $\epsilon_{\theta} = 0.3$): terminal functional error and the fraction of updates that violate \cref{eq:proj_prop}. (d) Sigma-modification gain: functional and parameter error. (e) Network width in S2 for the drift and adjoint laws.}
\label{fig:ablation_app}
\end{figure}

\begin{figure}[tbp]
\centering
\includegraphics[width=0.5\textwidth]{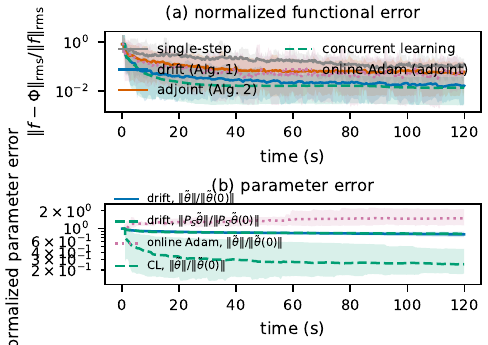}
\caption{S1, five seeds, windowed laws and baselines. (a) Normalized functional error. (b) Parameter error: drift law, total and projected onto the dominant excited subspace $S$; totals for CL and online Adam. The projected component contracts while the complementary component changes little; the total error of the unconstrained Adam flow grows.}
\label{fig:s1_subspace}
\end{figure}

\begin{figure}[tbp]
\centering
\begin{minipage}[t]{0.5\textwidth}
\centering
\includegraphics[width=\textwidth]{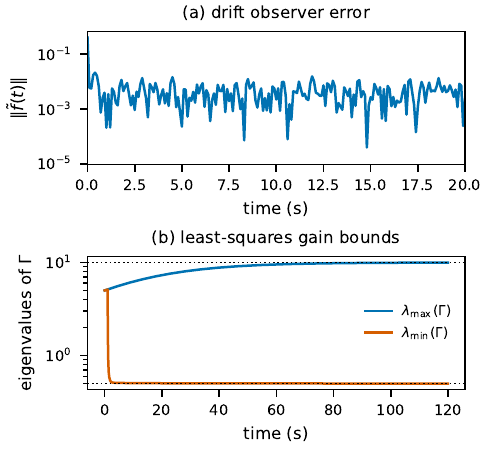}
\caption{S1, representative seed, sampled predictive-hold observer. (a) Drift observer error $\left\Vert \tilde{f}(t) \right\Vert$. (b) Extreme eigenvalues of $\Gamma(t)$, which remain within the certified interval $[\underline{\lambda}, \bar{\lambda}]$ of \cref{lem:gamma} in every seed of both studies.}
\label{fig:s1_internals}
\end{minipage}\hfill
\begin{minipage}[t]{0.46\textwidth}
\centering
\includegraphics[width=\textwidth]{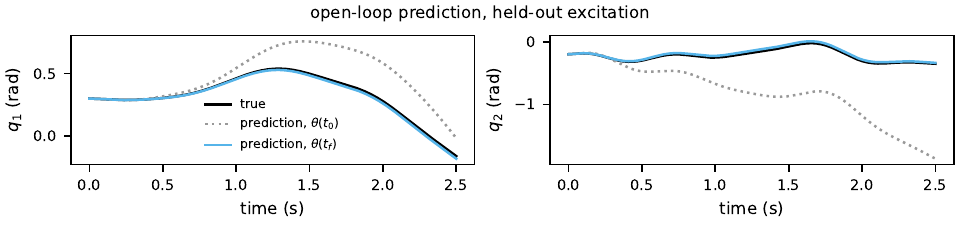}
\caption{S2, $2.5\,$s open-loop prediction on the held-out excitation under the initial and final drift-law parameters of seed $0$.}
\label{fig:s2_rollout}
\end{minipage}
\end{figure}

\begin{figure}[tbp]
\centering
\includegraphics[width=\textwidth]{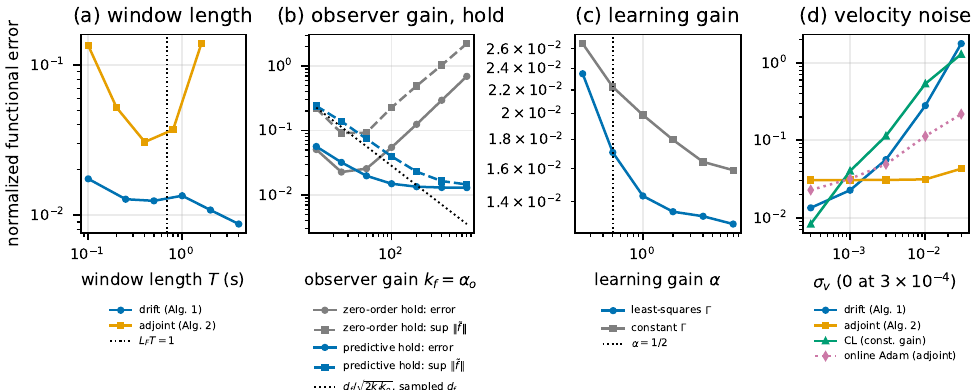}
\caption{Ablations of the windowed laws (S1; three seeds, $60\,$s; geometric means). (a) Window length, with $L_{F} T = 1$ at the measured $L_{F}$. (b) Observer gain under a zero-order and a predictive hold: terminal error (solid), post-transient supremum of $\left\Vert \tilde{f} \right\Vert$ (dashed), and the residual formula $d_{f}/\sqrt{2k_{f}k_{o}}$ evaluated with the sampled $d_{f}$ (dotted). (c) Learning gain with least-squares and constant gain matrices. (d) Velocity noise for the windowed laws, CL, and online Adam.}
\label{fig:ablation}
\end{figure}

\textbf{Ablations of the windowed laws.} \cref{fig:ablation} reports the ablations that bear on specific claims of the analysis. (a) The drift law improves slowly with longer windows, consistent with the linear growth of its perturbation constants in $T$, whereas the adjoint law improves up to $T = 0.4\,$s and then increases from $3.7 \times 10^{-2}$ to $1.4 \times 10^{-1}$ between $0.8\,$s and $1.6\,$s; the measured $L_{F} = 1.45$ places $L_{F} T = 1$ at $T = 0.69\,$s, where the exponential growth of the constants in \cref{eq:etaa_bound} predicts the degradation to begin, and the same scaling is why the stored-segment law uses segments of $0.05\,$s with a Gauss--Newton gain rather than long windows with a constant one. (b) Under a zero-order hold the observer gain traces a U-shaped tradeoff, since the residual formula $\bar{z}_{o}$ decreases as $1/k_{f}$ while the hold error, amplified by $k_{f}$, grows as $k_{f}\Delta t$; the predictive hold removes the first-order hold error, and the observed floor follows the $1/k_{f}$ scaling up to $k_{f} \approx 100$ before saturating at the second-order residual that the continuous-time bound does not model. (c) The functional error decreases monotonically in $\alpha$ with no instability below $\alpha = \tfrac{1}{2}$, so the condition is sufficient rather than necessary, and the least-squares gain improves on a constant gain at every $\alpha$. Removing the gate lowers the error by a further $35$ percent over $60\,$s but lets $\lambda_{\min}(\Gamma)$ collapse to $5 \times 10^{-3}$, so the certified rate is lost; a Euclidean projection composed with the matrix gain violates \cref{eq:proj_prop} in $1.9$ percent of updates on a tight parameter set with no change in functional error here; and leakage above $k_{\sigma} = 10^{-3}$ degrades both errors, the tradeoff that $k_{\sigma}\bar{\theta}^{2}$ in \cref{eq:cd_def} expresses (\cref{app:sim_details}).

\textbf{Additional ablations.} \cref{fig:ablation_app}(a) compares the gain-matrix variants. The ungated and pure least-squares updates attain terminal errors of $8.2 \times 10^{-3}$ and $8.9 \times 10^{-3}$ against $1.3 \times 10^{-2}$ for the gated update and $1.8 \times 10^{-2}$ for a constant gain, but their smallest eigenvalue falls to $5 \times 10^{-3}$ within $60\,$s and continues to decrease, so the certified rate in \cref{thm:drift} degrades without bound; the gate trades a modest loss in short-horizon accuracy for a uniform lower bound. \cref{fig:ablation_app}(b) shows that the terminal error and the excited dimension vary weakly and non-monotonically with the number of multisine components at equal total power (between $9.7 \times 10^{-3}$ and $1.8 \times 10^{-2}$, and between $6$ and $9$ dimensions), because fewer components at equal power produce larger excursions, which vary the Jacobian at least as much as additional frequencies do; the excited dimension is governed by the spatial extent of the trajectory rather than by the number of frequencies. \cref{fig:ablation_app}(c) reports the projection comparison discussed in \cref{sec:simulations}.

\textbf{Verification of the certified inequalities on simulated data.} The following checks were performed on the code and the simulated signals. (i) The projection inequality in \cref{eq:proj_prop} and the invariance of $\Theta_{\epsilon}$ were verified on $20{,}000$ random draws of $(\hat{\theta}, \theta^{*}, y, \Gamma)$, including $1{,}982$ draws in the active region and $5{,}000$ draws on the boundary, with no violations; on the simulated runs, the $\Gamma$-metric projection produced no violation in any update of any run. (ii) The gate $\varsigma$ implemented in code equals $0$ at $\lambda_{\min}(\Gamma) = \underline{\lambda}$, $1$ at $2\underline{\lambda}$, and is linear in between. (iii) The analytic Jacobians $\Phi'$ and $\partial\Phi/\partial X$ agree with central finite differences to $10^{-9}$, and the adjoint gradient agrees with central finite differences of the discretized \cref{eq:ET_def} with relative error $3 \times 10^{-4}$ (S1) and $5 \times 10^{-4}$ (S2) and cosine similarity $1.000000$. (iv) In every seed of both studies, $\lambda_{\min}(\Gamma(t)) \geq \underline{\lambda}$ and $\lambda_{\max}(\Gamma(t)) \leq \bar{\lambda}$ hold at every logged instant. (v) On the logged S1 seed-0 trajectory after $t = 5\,$s, the maximum sampled value of $\left\Vert \dot{f} \right\Vert$ is $2.83$. Substitution into the continuous-time residual formula gives $\bar{z}_{o} \approx 1.4 \times 10^{-2}$ at $k_{f} = 200$, while the sampled predictive-hold observer has a post-transient supremum of $\left\Vert \tilde{f} \right\Vert$ of $2.3 \times 10^{-2}$. The sampled maximum is a diagnostic estimate rather than a certified bound over the full trajectory, and the sampled observer includes a hold residual not covered by the continuous-time result, as \cref{fig:ablation}(b) shows across gains. (vi) At $t = 10\,$s in S1, the identity $\Xi_{f} = \Psi\tilde{\theta} + \eta_{f}$ was evaluated with $\tilde{\theta}$ known, and $\left\Vert \eta_{f} \right\Vert = 0.42$ lies within the bound of \cref{eq:etaf_bound} evaluated with the measured $\bar{\upsilon}$, $c_{\theta\theta}$, and $\bar{f}_{T}$, which equals $5.2$; the factor of twelve is the expected looseness of a Lagrange remainder bound. (vii) Two runs with the same seed produce bitwise identical parameter trajectories, and the plant trajectory is identical across the five methods for a given seed.

\end{document}